%% file: best-path.tex
\documentclass[10pt,conference,letterpaper]{IEEEtran}
\IEEEoverridecommandlockouts

\input{_default_packages.tex}

\usepackage{physics}
\usepackage{graphicx}
\usepackage{textcomp}
\usepackage{xcolor}

\usepackage{tikz}
\usetikzlibrary{shapes.geometric, arrows, shapes.symbols} 

\def\BibTeX{{\rm B\kern-.05em{\sc i\kern-.025em b}\kern-.08em
    T\kern-.1667em\lower.7ex\hbox{E}\kern-.125emX}}

    \newcommand{\oracle}{{\normalfont \texttt{BestPath}}}
    \newcommand{\bench}{{\normalfont \texttt{Bench}}\xspace}
    \newcommand{\rad}{{\normalfont \texttt{rad}}}
    
    \newcommand{\pl}{{\normalfont \text{(path)}}}
    \newcommand{\est}{{\normalfont \texttt{LinkEst}}}
    \newcommand{\Lmax}{{L_{\text{\normalfont max}}}}
    \newcommand{\Lmaxsquare}{{L_{\text{\normalfont max}}^2}}
    \newcommand{\bpilink}{{\normalfont \texttt{BeQuP-Link}}\xspace}
    \newcommand{\bpipath}{{\normalfont \texttt{BeQuP-Path}}\xspace}
    \newcommand{\bpi}{{\normalfont \texttt{BeQuP}}\xspace}
    \newcommand{\complexity}{Q}

\begin{document}

\title{Learning Best Paths in Quantum Networks\\
    \thanks{{\em Acknowledgments}---This research was supported in part by the NSF grant CNS-1955744, NSF-ERC Center for Quantum Networks grant
        EEC-1941583, and DOE Grant AK0000000018297.
        The work of John C.S. Lui was supported in part by the RGC SRFS2122-4S02.
        The work of Mohammad Hajiesmaili is supported by NSF CNS-2325956, CAREER-2045641, CPS-2136199, CNS-2102963, and CNS-2106299.
        The work of Xutong Liu was partially supported by a fellowship award from the Research Grants Council of the Hong Kong Special Administrative Region, China (CUHK PDFS2324-4S04).
        XW and ML thanks Dr. Yihan Du for providing implementation suggestions on link estimation.
        Maoli Liu and Xutong Liu are the corresponding authors.}
}

\author{
    \author{\IEEEauthorblockN{Xuchuang Wang$^1$, Maoli Liu$^2$, Xutong Liu$^{1,3}$, Zhuohua Li$^2$, Mohammad Hajiesmaili$^1$, John C.S. Lui$^2$, Don Towsley$^1$}
        \IEEEauthorblockA{$^1$\textit{University of Massachusetts, Amherst}, $^2$\textit{The Chinese University of Hong Kong}, $^3$\textit{Carnegie Mellon University}
            \\
            \{xuchuangwang, xutongliu, hajiesmaili, towsley\}@cs.umass.edu, \{mlliu, zhli, cslui\}@cse.cuhk.edu.hk}
    }
}

\maketitle

\begin{abstract}
    Quantum networks (QNs) transmit delicate quantum information across noisy quantum channels. Crucial applications, like quantum key distribution (QKD) and distributed quantum computation (DQC), rely on efficient quantum information transmission. Learning the best path between a pair of end nodes in a QN is key to enhancing such applications. This paper addresses learning the best path in a QN in the online learning setting. We explore two types of feedback: ``link-level'' and ``path-level''. Link-level feedback pertains to QNs with advanced quantum switches that enable link-level benchmarking. Path-level feedback, on the other hand, is associated with basic quantum switches that permit only path-level benchmarking. We introduce two online learning algorithms, BeQuP-Link and BeQuP-Path, to identify the best path using link-level and path-level feedback, respectively. To learn the best path, BeQuP-Link benchmarks the critical links dynamically, while BeQuP-Path relies on a subroutine, transferring path-level observations to estimate link-level parameters in a batch manner. We analyze the quantum resource complexity of these algorithms and demonstrate that both can efficiently and, with high probability, determine the best path. Finally, we perform NetSquid-based simulations and validate that both algorithms accurately and efficiently identify the best path.
\end{abstract}


\input{sections/intro.tex}
\input{sections/background.tex}
\input{sections/model.tex}

\input{sections/algo-link.tex}

\input{sections/algo-path.tex}

\input{sections/qkd-ext.tex}

\input{sections/simu.tex}

\input{sections/relat.tex}
\input{sections/concl.tex}

\bibliographystyle{IEEEtranN}
\bibliography{reference.bib}

\newpage
\onecolumn
\appendix

\input{sections/proof-link.tex}
\input{sections/proof-path.tex}

\end{document}

%% file: _default_packages.tex
\usepackage[numbers]{natbib}

\usepackage{graphicx}
\usepackage{diagbox}
\usepackage{multirow}
\usepackage{booktabs}
\usepackage{subcaption}
\usepackage{threeparttable, array, float}
\usepackage{colortbl}
\usepackage{tkz-graph}
\usetikzlibrary{calc}

\usepackage{textcomp}

\usepackage{xcolor}

\usepackage{hyperref}
\definecolor{mydarkblue}{rgb}{0,0.08,0.45}
\hypersetup{ %
    pdftitle={},
    pdfsubject={},
    pdfkeywords={},
    pdfborder=0 0 0,
    pdfpagemode=UseNone,
    colorlinks=true,
    linkcolor=mydarkblue,
    citecolor=mydarkblue,
    filecolor=mydarkblue,
    urlcolor=mydarkblue,
}

\usepackage{amsmath,amsfonts,amssymb}
\usepackage{bm, bbm}
\usepackage{mathtools}
\usepackage{nicefrac}

\makeatletter
\newcommand\addstarred[1]{%
    \expandafter\let\csname\string#1@nostar\endcsname#1%
    \edef#1{\noexpand\@ifstar\expandafter\noexpand\csname\string#1@star\endcsname\expandafter\noexpand\csname\string#1@nostar\endcsname}%
    \expandafter\newcommand\csname\string#1@star\endcsname%
}
\makeatother
\newcommand{\1}[1]{\mathbbm{1}{\{#1\}}}
\addstarred\1[1]{\mathbbm{1}{\left\{#1\right\}}}

\DeclarePairedDelimiter\ceil{\lceil}{\rceil}
\DeclarePairedDelimiter\floor{\lfloor}{\rfloor}
\DeclarePairedDelimiter\abs{\lvert}{\rvert}

\DeclareMathOperator*{\argmax}{arg\,max}

\renewcommand{\ge}{\geqslant}
\renewcommand{\le}{\leqslant}

\renewcommand{\P}{\mathbb{P}}

\newcommand{\upbra}[1]{^{(#1)}}

\newcommand{\type}[1]{Type-\uppercase\expandafter{\romannumeral#1}}

\usepackage{amsthm}
\newtheorem{theorem}{Theorem}
\newtheorem{lemma}[theorem]{Lemma}

\usepackage{algorithm}
\usepackage[noend]{algpseudocode}

\let\oldComment=\Comment
\renewcommand{\Comment}[1]{\oldComment{\texttt{#1}}}
\algnewcommand{\LeftComment}[1]{\Statex $\triangleright$ \texttt{#1}}
\algnewcommand{\RightComment}[1]{\Statex \leavevmode\hfill$\triangleright$ \texttt{#1}}

\algnewcommand\algorithmicinput{\textbf{Input:}}
\algnewcommand\Input{\item[\algorithmicinput]}%

\algnewcommand\algorithmicoutput{\textbf{Output:}}
\algnewcommand\Output{\item[\algorithmicoutput]}%

\algnewcommand\algorithmicinitial{\textbf{initialize:}}
\algnewcommand\Initial{\item[\algorithmicinitial]}%

\usepackage{xspace}
\usepackage{comment}

\usepackage{fontawesome5}

\usepackage[colorinlistoftodos,prependcaption,textsize=tiny,textwidth=1cm,color=green]{todonotes}


%% file: sections/intro.tex
\section{Introduction}\label{sec:intro}






Quantum networks (QNs) consist of quantum devices that exchange quantum information through quantum channels~\citep{wehner2018quantum}. These devices perform quantum operations on quantum states, such as quantum measurements and gates, to process quantum information. Quantum networks enable various applications, including quantum sensing~\citep{degen2017quantum}, quantum cryptography~\citep{bennet1984quantum}, and quantum computing~\citep{nielsen2001quantum}.

Direct long-distance quantum communication between two nodes is challenging due to the exponential decay rate of quantum information during transmission~\citep{wehner2018quantum}. To mitigate this decay, quantum repeaters~\citep{briegel1998quantum} have been proposed as intermediate nodes to split long-distance quantum communication into shorter segments.
Recently, \citet{dai2021entanglement} further extend the requirements of quantum repeaters to be able to connect multiple nodes, called \emph{quantum switches}.
Figure~\ref{fig:qn-illustration} illustrates a quantum network composed of quantum links and quantum switches.
We refer to direct connections between two quantum switches
as \emph{quantum links}, and connections involving intermediate nodes as \emph{quantum paths}, which consist of multiple quantum links.

However, even with quantum repeaters and switches, the no-cloning theorem of quantum mechanics~\citep{wootters1982single} prevents quantum information from being amplified, leading to noise and decoherence during transmission. To evaluate these errors, \emph{fidelity} has been proposed as a metric of the quantum channel~\citep{nielsen2001quantum}.


\begin{figure}[tp]
    \centering
    \includegraphics[width=0.47\textwidth]{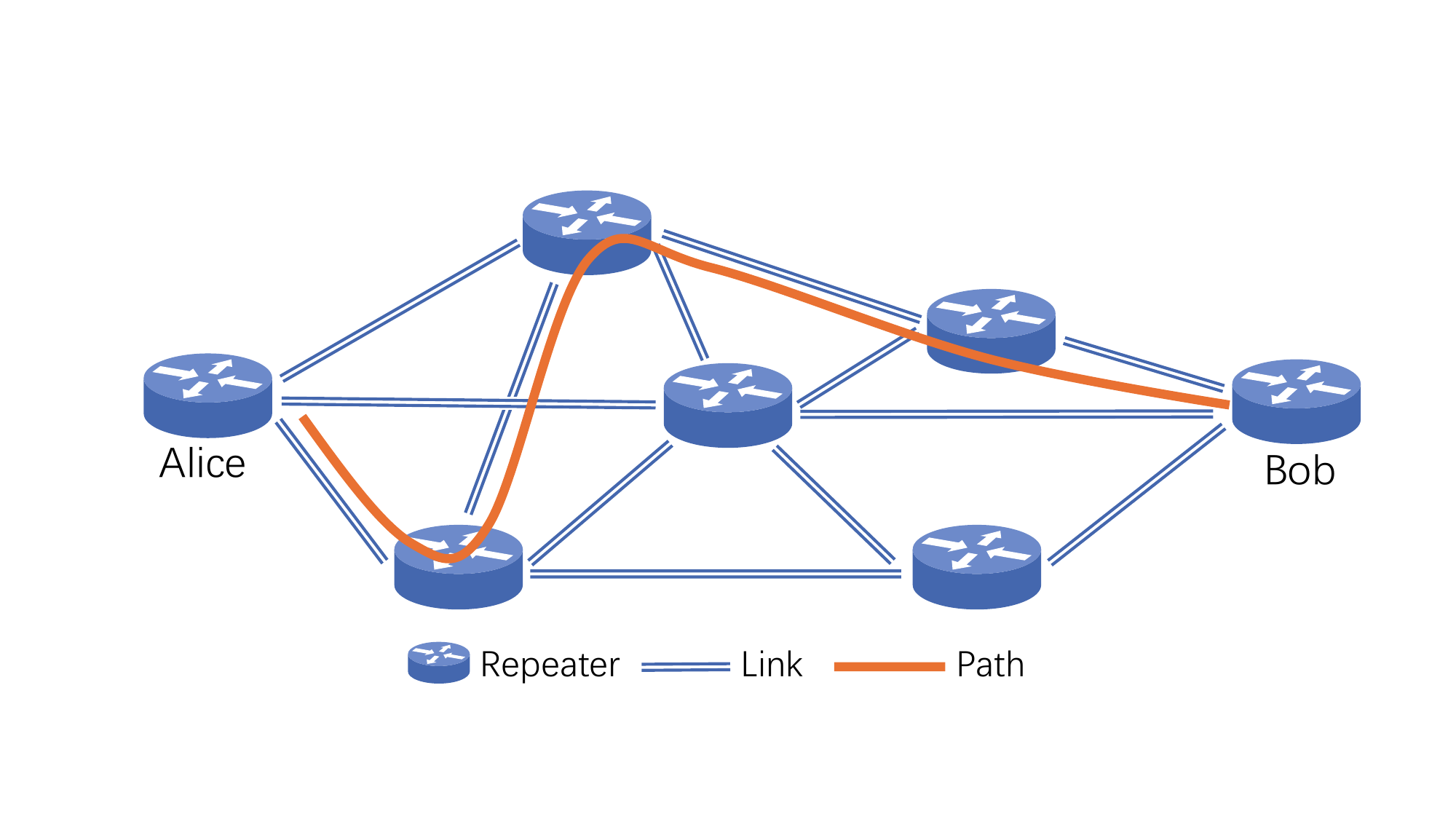}
    \caption{A quantum network with multiple links and paths
    }
    \label{fig:qn-illustration}
\end{figure}

This paper focuses on identifying the optimal path that maximizes fidelity between two nodes in a quantum network. Given a set of \(K\) paths and \(L\) links between two nodes, we first illustrate how to estimate the fidelity of a path using network benchmarking~\citep{helsen2023benchmarking}. To estimate these parameters, we introduce two types of quantum network benchmarking: (1) link-level and (2) path-level. Link-level benchmarking evaluates fidelity parameters of individual quantum links and is suitable for quantum networks whose switches provide access to quantum operations. Path-level benchmarking evaluates the fidelity parameters of quantum paths and is suitable for networks that do not provide access to switches. We define this problem as learning the \underline{Be}st \underline{Qu}antum \underline{P}ath (\bpi).

In a link-level feedback environment, a naive approach would benchmark all links in the quantum network to estimate their fidelity parameters and identify the highest fidelity path. However, this approach is resource-intensive, especially if multiple paths share links. For example, if some links are shared by all paths, benchmarking only the unique links in each path would suffice to identify the best path, making benchmarking overlapped links redundant.
To address this, we propose the \bpilink algorithm, which sequentially benchmarks unique links in chosen paths to reduce resource costs while ensuring the identification of the best path with high probability. We also derive an upper bound on the quantum resource requirement of this algorithm.

In a path-level feedback environment, the potentially combinatorial explosion in the number of paths presents a challenge. To address this, we propose an estimation procedure that maps path parameter estimates to link parameter estimates via linear regression. Using these estimates, we introduce the \bpipath algorithm, which maintains and successively reduces a candidate set of paths until only the best path remains. This algorithm identifies the best path with high confidence
and low quantum resource costs.

The contributions of this paper are as follows:

\begin{itemize}
    \item We propose a quantum network model, \bpi, featuring a benchmarking subroutine to identify the path with the highest fidelity under both link-level and path-level benchmarking (Section~\ref{sec:model}).
    \item We introduce two algorithms, \bpilink for link-level benchmarking (Section~\ref{sec:link}) and \bpipath for path-level benchmarking (Section~\ref{sec:path}). Both algorithms identify the best path with high probability, and we analyze their quantum resource complexity requirements.
    \item Besides the fidelity metric for quantum communication, another critical application of QNs is Quantum Key Distribution (QKD)~\citep{bennet1984quantum}.
          The efficiency of a key distribution protocol are quantified by the secret key fraction (SKF).
          We extend our model and algorithms to identify the best path with the highest SKF in Section~\ref{sec:qkd-extension}.
    \item We validate the performance of the proposed algorithms and compare them with existing baselines using the quantum network simulator NetSquid~\citep{coopmans2021netsquid} (Section~\ref{sec:simulation}).
\end{itemize}


%% file: sections/background.tex
\section{Background}
\label{sec:background}

\subsection{Quantum Network}
\label{subsec:quantum-network}

Quantum networks transmit quantum states between nodes, playing a critical role in technologies like QKD~\citep{bennet1984quantum}, quantum sensing~\citep{degen2017quantum}, and quantum computing~\citep{nielsen2001quantum}. However, long-distance transmission is challenging due to communication rate decay~\citep{wehner2018quantum}. To overcome this, networks use a hierarchical structure with shorter segments connected by quantum repeaters~\citep{briegel1998quantum} and switches~\citep{dai2021entanglement}.
Quantum networks face three primary noise (error) sources: (1) loss errors during quantum state transmission, (2) operation errors during quantum gate operations~\citep{muralidharan2016optimal}, and (3) quantum memory decoherence, where we focus on the first two errors in this work.
To mitigate these errors, two main approaches are used: heralded entanglement generation (HEG)~\citep{briegel1998quantum}, and quantum error correction (QEC)~\citep{munro2012quantum}, where the HEG requires assistance of two-way classical communication,
while the QEC requires only one-way classical communication.

When both errors are addressed by QEC, it is called a \emph{one-way} quantum network.
When both error sources are addressed by HEG, the network is termed a \emph{two-way} quantum network.
In the  literature~\citep{muralidharan2016optimal}, a two-way quantum network is also referred as a first-generation quantum network, while a one-way quantum network is referred to as a third-generation quantum network. A network that combines QEC and HEG approaches is known as a {second-generation} quantum network.


\subsection{Network Benchmarking} \label{subsec:network-benchmark}
One of the key challenges in building quantum networks is that quantum information is vulnerable to noise sources that the environment, devices, and eavesdroppers could cause. Noise is often modeled as a quantum channel \(\Lambda\) that transforms an input quantum state into another output quantum state~\citep[Chapter 8]{nielsen2001quantum}. We note that the quantum channel is mainly a mathematical model, and there can be different physical realizations of the quantum channel, e.g., one-way or two-way quantum networks.
Several metrics have been introduced to characterize the noise in the quantum channel; we will use \emph{fidelity}~\citep{nielsen2001quantum}, a measure of how well the output quantum state \(\Lambda(\ket\psi\bra\psi)\) approximates the input quantum state \(\ket\psi\bra\psi\). Formally, we define the average fidelity of the channel \(\Lambda\) as \(f \coloneqq \int \Tr[\Lambda(\ket\psi\bra\psi)\ket\psi\bra\psi]  d\psi\) where the integral is taken over all pure quantum states \(\ket \psi\).


To estimate the fidelity of the quantum channel, one can use \emph{network benchmarking}~\citep{helsen2023benchmarking}.
Network benchmarking operates on the principle that a channel twirling process~\citep{magesan2012characterizing}, which involves the random application of Clifford operations~\cite{clifford1998the}, converts any quantum channel into a depolarizing channel that maintains the same level of fidelity. By repeatedly accessing these depolarizing channels, their average fidelity can be assessed, effectively reflecting the fidelity of the original channel. This process utilizes parameters $\mathcal{M}$ and $T_0$, where $\mathcal{M}$, known as the bounce number set, comprises a sequence of integers, and $T_0$ indicates the number of repetitions for each bounce number $m \in \mathcal{M}$. A ``bounce'' occurs when node $S$ applies a random Clifford operation to the state and forwards it to node $D$, which reciprocates before sending it back to $S$.

To calculate the average fidelity of the channel
between nodes $S$ and $D$, this sequence is executed $T_0$ times for each $m \in \mathcal{M}$, following these steps: (i) node $S$ generates an initial state; (ii) nodes $S$ and $D$ bounce the state $m$ times; (iii) node $S$ performs a final operation and measures the state. The mean of the $T_0$ measurement outcomes, represented as $b_m$, often referred to as the survival probability, is determined by the exponential model $b_m = Ap^{2m}$, where $A$ is a constant that compensates for errors in quantum state preparation and measurement, and $p$ is the depolarizing parameter of the twirled channel.
From the depolarizing parameter $p$, one can infer the channel fidelity as $f=(1+{p})/2$~\citep{helsen2023benchmarking}.
Therefore, by applying this exponential model $b_m = Ap^{2m}$ to the collected network benchmarking data,
one can estimate the depolarizing parameter \(p\), denoted as $\hat{p}$, and infer the channel fidelity \(\hat f = (1+\hat{p})/2\).
We denote the network benchmarking over nodes \(S\) and \(D\) (denoted as \(x_{S\leftrightarrow D}\)) under \(T_0\) repetitions and \(\mathcal M\) bounce number set as \(\bench(x_{S\leftrightarrow D}; T_0, \mathcal{M})\).

The benchmarking process requires both $S$ and $D$ to be able to perform the needed quantum operations and measurements.
This requirement can be supported by most nodes (e.g., quantum switches) in the one-way quantum network, but only by source and destination nodes in the two-way quantum network, where most internal nodes are often simple repeaters.
Later, we model the former one with the ability to benchmark each link as \emph{link-level} benchmarking and the latter as \emph{path-level} benchmarking.
Denote the cost of benchmarking one link as a unit of quantum resources.
Then, benchmarking a path consisting of \(x\) links consumes \(x\) units of quantum resources since this benchmarking incurs cost across all links in the path.

%% file: sections/model.tex
\section{Model}\label{sec:model}

\subsection{Network Model Parameters}

Consider one pair of source and destination nodes connected by multiple paths in a quantum network (as in Figure~\ref{fig:qn-illustration}). Let \(K\in\mathbb{N}^+\) denote the total number of quantum paths between these two nodes, forming a set of paths \(\mathcal{K}\coloneqq \{1,2,\dots, K\}\).
Let \(L \in \mathbb{N}^+\) denote the total number of quantum links composing these \(K\) paths and let \(\mathcal{L}\coloneqq \{1,2,\dots, L\}\) be the set of links.
For any path \(k\in\mathcal{K}\), let \(\bm x (k) \coloneqq (x_\ell (k))_{\ell\in\mathcal{L}} \in \{0,1\}^L\) be a binary vector representing whether the link \(\ell\) belongs to path \(k\), where \(x_\ell (k) = 1\) if link \(\ell\) is in path \(k\), and \(x_\ell (k) = 0\) otherwise.
Let \(L(k) \coloneqq \sum_{\ell \in\mathcal{L}} x_\ell (k)\) denote  the number of quantum links in path \(k\), and \(\mathcal{L}(k) \coloneqq \{\ell\in\mathcal{L}: x_\ell (k) = 1\}\) as the set of quantum links in path \(k\).
Without loss of generality, assume the matrix \(\bm X \coloneqq \{\bm x(1), \bm x(2), \dots, \bm x(K)\}\) has a rank of exactly \(L\). If not, we can project the matrix to a lower-dimensional space where the corresponding network graph retains the same best arm.

Each quantum link \(\ell\in \mathcal{L}\) is associated with fidelity parameter \(f_\ell \in (0,1)\) and depolarizing parameter \(p_\ell\in (0,1)\).
We use the boldface letter, e.g., \(\bm p = (p_\ell)_{\ell\in\mathcal{L}}\), to denote a vector with \(L\) entries, each entry corresponding to one of the \(L\) links.
Similarly, each path \(k\in\mathcal{K}\) is associated with fidelity parameter \(f^\pl(k)\) and depolarizing parameter \(p^\pl(k)\).
The depolarizing parameter \(p^\pl(k)\) equals the product of the depolarizing parameter of each quantum link in the path~\citep{vardoyan2023quantum,helsen2023benchmarking}, that is,
\begin{align}
    p^\pl(k) & = \prod_{\ell \in \mathcal{L}(k)} p_\ell.\label{eq:depolarizing-parameter-product-relation}
\end{align}
In this model, the network topology is known a priori, but the link and path parameters, e.g., \(f_\ell, f^{\pl}(k)\), are unknown, and one needs to benchmark (learn) these parameters.

\subsection{Fidelity parameter transformation}

Our goal is to identify the best path, the one with the highest channel fidelity \(f^{\pl}(k)\) in the quantum network, that is,
\({
        k^* = \argmax_{k\in\mathcal{K}} f^{\pl}(k),
    }\)
and without loss of generality, we assume that the best path is unique.
We first elaborate on a monotonic bijection mapping of the fidelity \(f^{\pl}\) as follows,
\begin{align}
    f^{\pl}(k) \overset{(a)}\longleftrightarrow p^{\pl}(k) \overset{(b)}\longleftrightarrow \sum_{\ell \in \mathcal{L}(k)} \log p_\ell \coloneqq F(k),
    \label{eq:fidelity-transformation}
\end{align}
where mapping (a) is due to the relation between the fidelity and the depolarizing parameter of the quantum path of network benchmarking (Section~\ref{subsec:network-benchmark}), i.e., \(f^{\pl}(k) = {(p^{\pl}(k) + 1)}/{2}\),
mapping (b) is by taking $\log$ on~\eqref{eq:depolarizing-parameter-product-relation},
and we define \(F(k)\) as the transformed fidelity since it has a monotonic bijection with the original fidelity \(f^{\pl}(k)\).

For ease of presentation, we slightly abuse the notation of scalar \(\log\) function to represent a entry-wise vector \(\log\) function, such that \(\log \bm p = (\log p_\ell)_{\ell \in \mathcal{L}}\).
Therefore, the transformed fidelity function \(F(k)\) can be expressed as an inner product of the binary vector \(\bm x(k)\) for path \(k\) and the vector \(\log \bm p\), that is, \(F(k) = \bm x^T(k) \log \bm p\).
As \(\argmax_{k\in\mathcal{K}}f^{\pl}(k) = \argmax_{k\in\mathcal{K}}F(k)\), the best path \(k^*\) can also be identified by maximizing \(F(k)\) in~\eqref{eq:fidelity-transformation}.
Hereafter, we estimate the depolarizing parameter \(p_\ell\) of each quantum link in the network and use \(p_\ell\) to calculate fidelity when necessary.

To facilitate later analysis of our algorithms, we define two fidelity gaps.
For link-level, we define the gap as follows,
\[
    \Delta_\ell \coloneqq \begin{cases}
        F(k^*) - \max_{k \in \mathcal{K}: \ell \in \mathcal{L}(k)} F(k)    & \text{if } \ell \notin \mathcal{L}(k^*) \\
        F(k^*) - \max_{k \in \mathcal{K}: \ell \notin \mathcal{L}(k)} F(k) & \text{if } \ell \in \mathcal{L}(k^*).
    \end{cases}
\]
The link gap \(\Delta_\ell\) measures the difference in fidelities of the best global path \(k^*\) and the best local path among the paths that (do not) contain link \(\ell\) when the link (does not) belongs to the optimal path.
For path-level, we define the path gap for any suboptimal path \(k\neq k^*\),
\(
\Delta^\pl(k) \coloneqq F(k^*) - F(k),
\)
which is the difference in fidelities of the best path \(k^*\) and suboptimal path $k\ne k^*$; and for the best path \(k^*\), the gap is \(\Delta^\pl(k^*) \coloneqq \min_{k\neq k^*}\Delta^\pl(k)\).

\subsection{Link-level and Path-level Network Benchmarking}

In Section~\ref{subsec:network-benchmark}, we define the network benchmarking subroutine as \(\bench(x_{S\leftrightarrow D}; T_0, \mathcal{M})\), which estimates the depolarizing parameter \(p_{S\leftrightarrow D}\) of the channel between nodes \(S\) and \(D\), with both nodes needing the ability to perform advanced quantum operations and measurements. For simplicity, we fix the bouncing number set \(\mathcal{M}\) and rounds \(T_0\) and focus on choosing the benchmarking channel \(x\). We simplify the notation as \(\bench(x)\), where \(x\) is the benchmarked channel.

Depending on whether intermediate quantum repeaters can perform the necessary quantum operations, we consider two types of benchmarking applications: link-level and path-level benchmarking. \emph{Link-level benchmarking} applies \(\bench(x)\) to a quantum link, i.e., \(x=\ell\) for any link \(\ell \in \mathcal{L}\), to estimate the link depolarizing parameter \(p_\ell\). \emph{Path-level benchmarking} (a.k.a., end-to-end) applies \(\bench(x)\) to a quantum path, i.e., \(x=\mathcal{L}(k)\) for any path \(k\) in the network, to estimate the path depolarizing parameter \(p^\pl(k)\). As discussed in Section~\ref{subsec:network-benchmark}, path-level benchmarking suits quantum networks with limited quantum processing capabilities, while link-level benchmarking is better for networks with advanced quantum processing capabilities. One link-level query incurs one unit of quantum resource cost, whereas one path-level query on path \(k\) incurs \(L(k)\) units of quantum resource cost since the path-level query relies on \(L(k)\) established links in the path. We then state results for network benchmarking.


\begin{lemma}[{Adapted from~\citet[Lemma 1]{liu2024link}}]\label{lma:benchmarking-concentration}
    For any quantum link \(\ell \in \mathcal{L}\), given the average \(\hat p_\ell\) of \(N\in \mathbb{N}^+\) samples of \(\bench(\ell; \mathcal{M}, T_0)\), and parameter \({\delta\in (0,1)}\), we have
    \begin{equation}
        \label{eq:bench-concentration}
        \P(
        \abs{\hat p_\ell - p_\ell}
        \le \sqrt{{C\log \left(\delta^{-1}\right)}/{N}}
        )
        \ge 1- \delta,
    \end{equation}
    where the constant \(C\) depends on the bouncing number set \(\mathcal{M}\), the rounds \(T_0\), and other network parameters.
\end{lemma}

Lemma~\ref{lma:benchmarking-concentration} shows that benchmarking can estimate the depolarizing parameter of a link with high confidence, and~\eqref{eq:bench-concentration} depicts the concentration rate of the estimation. A similar concentration result holds for path-level benchmarking.

\subsection{Online Learning Problem Formulation}

We consider the problem of learning the best path in the quantum network: a sequential decision process, where, at each time \(t\), the learner selects a link \(\ell_t\) (or path \(k_t\)) to benchmark and receives estimates of link (or path) depolarizing parameters \(\hat p_{\ell_t}\) (or \(\hat p^{\pl}(k_t)\)) from the network benchmarking subroutine.
From the benchmarking feedback, the learner updates the estimate of the depolarizing parameter of each link (or path) and then decides which to benchmark in the next slot. We summarize the problem in Procedure~\ref{proc:learning-best-path}.

Given a confidence parameter \(\delta\in (0,1)\), we aim to identify the best path \(k^*\) with confidence \(1-\delta\), i.e., \(\mathbb{P}(\hat k_{\text{output}} = k^*) \ge 1-\delta\), with as few costs of
quantum resources as possible. We evaluate the learning algorithm's performance by the total quantum resources consumed,
which is called (quantum) resource complexity, denoted by \(Q\).
Let \(T\) denote the last round of the algorithm (stopping round). Then, \(\complexity = T\) under link-level benchmarking and \(\complexity = \sum_{t=1}^T L(k_t)\) under path-level benchmarking.

\floatname{algorithm}{Procedure}
\begin{algorithm}[tb]
    \caption{\bpi: Learning \underline{Be}st \underline{Qu}antum \underline{P}ath}\label{proc:learning-best-path}
    \begin{algorithmic}[1]
        \Input Path set \(\mathcal{K}\), link set \(\mathcal{L}\),
        confidence parameter \(\delta\)
        \Repeat
        \State Select a link \(\ell_t\) / path \(k_t\) to benchmark
        \State Observe link-level feedback \(X_{\ell_t,t}\) from subroutine
        \(\bench(\ell_t)\) / path-level \(Y_t(k_t)\)  from \(\bench(\mathcal{L}(k_t))\)

        \Until{Identify the best path \(k^*\) with confidence \(1-\delta\)}
        \Output Best path \(k^*\)
    \end{algorithmic}
\end{algorithm}
\floatname{algorithm}{Algorithm}

%% file: sections/algo-link.tex
\section{A Link-Level Algorithm} \label{sec:link}

This section presents the \bpilink algorithm to identify the best path in a quantum network with link-level feedback and the algorithm's resource complexity analysis.

\begin{algorithm}[tp]
    \caption{\bpilink: {\small Link-Level Best Path Identification}}
    \label{alg:link-level}
    \begin{algorithmic}[1]
        \Input Path set \(\mathcal{K}\), link set \(\mathcal{L}\),
        confidence parameter \(\delta\), confidence radius function \(\rad_t(N)\),  subroutine \(\oracle\)

        \State \textbf{initialization:}
        set time \(t\gets L\);
        benchmark each link \(\ell\in\mathcal{L}\) once to initialize its depolarizing parameter estimate \(\hat{p}_{\ell, t}\), set corresponding counter \(N_{\ell, t}\gets 1\)

        \While{\texttt{true}}

        \State \(\hat k_t \gets \oracle(\hat{\bm{p}}_{t}, \mathcal{K})\) \label{line:empirical-best-path}
        \RightComment{input neutral link estimates}

        \State \(\tilde{p}_{\ell, t} \gets \begin{cases}
            \hat{p}_{\ell, t} - \rad_t(N_{\ell, t}) & \text{if } \ell \in \mathcal{L}(\hat k_t) \\
            \hat{p}_{\ell, t} + \rad_t(N_{\ell, t}) & \text{otherwise}
        \end{cases}\) \label{line:confidence-estimate}
        \RightComment{pessimistic/optimistic link estimates}

        \State \(\tilde{k}_{t} \gets \oracle(\tilde{\bm{p}}_{t}, \mathcal{K})\)
        \Comment{exploration}\label{line:optimistic-best-path}

        \If {\(\mathcal{L}(\hat k_t) = \mathcal{L}(\tilde k_t)\)} \label{line:break-condition}
        \State \textbf{break} \Comment{identified the best path}
        \EndIf

        \State \(\ell_t \gets \argmax_{\ell \in \mathcal{L}(\hat k_t) \bigtriangleup \mathcal{L}(\tilde k_t)} \rad_t(N_{\ell, t}))\) \footnotemark
        \RightComment{break ties arbitrarily}\label{line:link-selection}

        \State \(X_{\ell_t, t} \gets \bench(\ell_t)\) \label{line:link-benchmark}

        \State \(N_{\ell, t+ 1} \gets
        \begin{cases}
            N_{\ell, t} + 1 & \text{if } \ell = \ell_t \\
            N_{\ell, t}     & \text{otherwise}
        \end{cases}\) \label{line:counter-update}

        \State \(\hat{p}_{\ell, t+1} \gets
        \begin{cases}
            (X_{\ell, t} + N_{\ell, t} \hat{p}_{\ell, t}) / N_{\ell, t+1} & \text{if } \ell = \ell_t \\
            \hat{p}_{\ell, t}                                             & \text{otherwise}
        \end{cases}
        \)\label{line:link-estimate-update}

        \State \(t\gets t+1\) \label{line:time-update}

        \EndWhile

        \Output Path \(\hat k_t\)
    \end{algorithmic}
\end{algorithm}
\footnotetext{For two sets \(\mathcal A\) and \(\mathcal B\), the operator \(\mathcal{A} \bigtriangleup \mathcal{B}\) is the symmetric difference of the two sets defined as follows, \(\mathcal{A} \bigtriangleup \mathcal{B} \coloneqq (\mathcal{A}\setminus \mathcal{B}) \cup (\mathcal{B} \setminus \mathcal{A})\), where \(\mathcal{A}\setminus \mathcal{B}\coloneqq \mathcal{A}\cap \mathcal{B}^C\) is the set minus operation.}

\subsection{Algorithm Design}


\subsubsection{\oracle{} Subroutine}
Given the link depolarizing parameters \({\bm p = (p_\ell)_{\ell\in\mathcal{L}}}\) and the path set \(\mathcal{S}\subseteq\mathcal{K}\), \oracle{} outputs the best path among set \(\mathcal S\) that maximizes fidelity function \(F(k)\), that is, \(\oracle: (\bm p, \mathcal{S}) \to \argmax_{k\in \mathcal{S}} F(k)\).
For example, when input path set \(\mathcal{S}=\mathcal{K}\), one can implement \oracle{} by applying Dijkstra's algorithm~\citep{dijkstra1959note} to find the shortest path in the quantum network graph with edge lengths (or link weights) \(-\log (p_{\ell})\).
To estimate \(\{\log (p_{\ell})\}_{\ell\in\mathcal{L}}\) and solve the \oracle{} function, we need to transfer the estimates of depolarizing parameters \(\hat{p}_{\ell}\) to the estimate \({\log \hat p_\ell}\),
which introduces additional estimation errors to our algorithm. Later in~\ref{subsec:link-analysis}, we analyze the impact of these errors on the algorithm's resource complexity performance.




\subsubsection{\bpilink Algorithm}
\bpilink (Algorithm~\ref{alg:link-level}) takes the path set \(\mathcal K\), link set \(\mathcal L\), confidence parameter \(\delta\), radius function \(\rad_t(\cdot)\) (defined later in Theorem~\ref{thm:path-complexity}), and subroutine \oracle{} as input
and returns the best path \(k^*\) as output.
During each time slot \(t\), the algorithm first calls \oracle{} to return the best path \(\hat k_t\) based on the current link estimates \(\hat{\bm{p}}_t\).
Then, the algorithm constructs confidence estimates \(\tilde p_{\ell,t}\) for all link depolarizing parameters: use optimistic estimates \(\tilde{p}_{\ell,t} = \hat{p}_{\ell,t} + \rad_t(N_{\ell, t})\) for links \emph{not} in the empirical best path \(\hat k_t\), and pessimistic estimates \(\tilde{p}_{\ell,t} = \hat{p}_{\ell,t} - \rad_t(N_{\ell, t})\) for links in the empirical best path \(\hat k_t\) (Line~\ref{line:confidence-estimate}), where the \(\rad_t(N_{\ell, t})\) is the confidence radius of the link parameter estimate \(\hat p_{\ell, t}\) at time \(t\).
Then, based on confidence estimates \(\tilde{\bm p}_t\), \bpilink calls \oracle{} to calculate the second best empirical path \(\tilde k_t\)  (Line~\ref{line:optimistic-best-path}).
All links favored in the empirical best path \(\mathcal{L}(\hat k_t)\) are pessimistically estimated, while all other links in disfavor (not in \(\mathcal{L}(\hat k_t)\)) are optimistically estimated.
Therefore, if another path is potentially better than the first empirically best path \(\hat k_t\), it would be identified as the second empirical best path \(\tilde k_t\).

The algorithm then checks if the empirical best path \(\hat k_t\) and the exploratory best path \(\tilde k_t\) are the same.
If so, no other arm is potentially better,
and the algorithm terminates and outputs the empirical best path \(\hat k_t\) (Line~\ref{line:break-condition}), which is the best path with high probability.
Otherwise, the algorithm continues as follows. 
Among the links in the symmetric difference \(\mathcal{L}(\hat k_t)\bigtriangleup \mathcal{L}(\tilde k_t)\), i.e., either exclusively in path \(\hat k_t\) or exclusively in path \(\tilde k_t\), the algorithm selects the link with the largest confidence radius \(\rad_t(N_{\ell, t})\) (i.e., uncertainty) to benchmark (Line~\ref{line:link-selection}).
After receiving a new observation from the chosen link \(\ell_t\) (Line~\ref{line:link-benchmark}), the algorithm updates the depolarizing parameter estimate \(\hat p_{\ell,t}\) and the counter \(N_{\ell,t}\) for the link \(\ell_t\) (Lines~\ref{line:counter-update} and \ref{line:link-estimate-update}).

\subsection{Quantum Resource Complexity Analysis}\label{subsec:link-analysis}

This subsection presents the resource complexity of Algorithm~\ref{alg:link-level} for identifying the best path with link-level feedback.
The detailed proofs of all theoretical results of this paper will be provided in an extended version due to the space limit.

\begin{theorem}[Resource complexity of \bpilink]\label{thm:link-complexity}
    Given confidence parameter \(\delta > 0\) and set the radius function \(\rad_t(N_{\ell, t}) = \sqrt{{C\log \left( \frac{2Lt^3}{\delta} \right)}/{N_{\ell,t}}}\) with the same constant \(C\) in Lemma~\ref{lma:benchmarking-concentration},
    Algorithm~\ref{alg:link-level} identifies the best path with probability at least \(1 - \delta\) before
    \begin{equation}
        \label{eq:complexity-link-level}
        \complexity^{(\text{\emph{link}})}=O\left( L_{\text{\normalfont max}}^2 \sum_{\ell \in \mathcal{L}} \frac{1}{\Delta_\ell^2} \log \left( \frac{L}{\delta}\sum_{\ell \in \mathcal{L}} \frac{1}{\Delta_\ell^2} \right) \right),
    \end{equation}
    where \(\Lmax \coloneqq \max_{k\in\mathcal{K}} L(k)\) is the length of the longest path.
\end{theorem}


The resource complexity of Algorithm~\ref{alg:link-level} has a quadratic dependence on the length of the longest path \(\Lmax\).
This \(\Lmax\) dependence is due to the estimation error introduced by the transformation from the link depolarizing parameter \(\hat p_{\ell}\) to the path depolarizing parameter \(\hat p^{\pl} (k)\) in the \oracle{} function.
In practice, the depolarizing parameter \(p_\ell\in (0,1)\) often has a tighter range, e.g., between \((p_{\text{min}}, p_{\text{max}})\).
Because less than \(p_{\text{min}}\) indicates that the link is too noisy to be useful for transmitting quantum information, and greater than \(p_{\text{max}}\) requires advanced but unaffordable physical implementations.
With this constraint, one can show that any path whose length is greater than \(L_{\text{min}} \log_{1/p_{\text{max}}} (1/p_{\text{min}})\) cannot be the best path, where \(L_{\text{min}}\) is the length of the shortest path, and hence, we can remove these paths and say that \(L_{\text{max}} \le L_{\text{min}} \log_{1/p_{\text{max}}} (1/p_{\text{min}})\), which is often much smaller than \(L\).
Therefore, the other term \(\sum_{\ell \in \mathcal{L}}(1/\Delta_\ell^{2})\), linear in the number of links \(L\), is the dominant term in the resource complexity of Algorithm~\ref{alg:link-level}, which is avoidable according to the best arm identification literature~\citep{chen2014combinatorial}.



%% file: sections/algo-path.tex
\section{A Path-Level Algorithm}\label{sec:path}

In this section, we present the \bpipath algorithm to identify the best path in a quantum network with path-level feedback and determine its resource complexity.

\subsection{Algorithm Design}\label{subsec:path-algorithm}
\begin{algorithm}[tb]
    \caption{\(\est(\mathcal{S}, N)\): Estimate Link Depolarizing Parameter from Path-level Benchmarking}\label{alg:link-estimate}
    \begin{algorithmic}[1]

        \Input A set of path indices \(\mathcal{S}\) and a number of samples \(N\)

        \State \(\bm \lambda\upbra{\mathcal{S}} \gets \argmax_{\bm\lambda} \mathbb{E}_{k \sim \mathcal{D}(\bm\lambda, \mathcal{S})} [\bm x(k) \bm x^T(k)]\)
        \label{line:expermental-design}
        \For{\(n=1,\dots, N\)}\label{line:sample-loop}
        \State Sample path \(k_n\) from \(\mathcal{S}\) with probability \(\mathcal{D}(\bm \lambda\upbra{\mathcal{S}}, \mathcal{S})\)
        \State \(Y_{n} \gets \bench(\mathcal{L}(k_n))\)
        \EndFor\label{line:sample-loop-end}
        \State \(\bm A \gets N \sum_{k \in \mathcal{S}} \lambda\upbra{\mathcal{S}}(k) \bm x(k) \bm x^T(k)\)
        \label{line:linear-regression-start}
        \State \(\bm b \gets \sum_{n=1}^N \log(Y_n) \bm x(k_n)\)
        \vspace{1.5pt}
        \Output \({\log \hat{\bm p}} \gets \bm A^{-1}\bm{b}\) \label{line:linear-regression}
        \label{line:return-link-estimate}
    \end{algorithmic}
\end{algorithm}


Under the path-level feedback, we benchmark each path \(k \in \mathcal{K}\) and estimate depolarizing parameter \(p^\pl(k)\) for this path. Although one can regard each path as an individual quantum ``link'' and apply the LinkSelFiE~\cite{liu2024link} algorithm to identify the best one, the algorithm requires a large number of resources that has a linear dependence on the number of paths \(K\), as large as \(O\left( L^{L_{\text{max}}} \right)\) in the worst case. That is, individually learning each single path incurs a resource complexity that grows exponentially in the length of the longest path \(\Lmax\), which is impractical for large quantum networks. To avoid this exponential dependence, one should take advantage of overlaps among paths, i.e., that paths may share links. Thus, we devise a subroutine, called~\est{}, to estimate the depolarizing parameter of each link from the path-level feedback. In this subsection, we start by presenting the details of~\est{}, and then present the main algorithm to identify the best path using path-level feedback.



\subsubsection{\est{} Subroutine}
The \est{} subroutine has two inputs: Path set \(\mathcal{S}\) and the number of samples \(N\).
Noting the multiplicative relation between link-level and path-level parameters, \(p^\pl(k) = \prod_{\ell \in \mathcal{L}(k)} p_\ell\), we can estimate the depolarizing parameter of each link by solving a linear regression, composed of equations obtained from applying the \(\log\) function to both sides of the equation, i.e., \(\log p^\pl(k) = \sum_{\ell\in\mathcal{L}(k)} \log p_\ell\).

\est{} starts by collecting path-level samples, following the G-optimal design principle~\citep{tao2018best}, which aims to maximize estimation accuracy. That is, sampling paths according to a distribution that maximizes the variance matrix \(\mathbb{E}_{k \sim \mathcal{D}(\bm\lambda, \mathcal{S})} [\bm x(k) \bm x^T(k)]\) where \(\mathcal{D}(\bm\lambda, \mathcal{S})\) is a multinomial distribution over paths in \(\mathcal{S}\) with discrete probabilities in vector \(\bm\lambda\in [0,1]^{\abs{\mathcal{S}}}\) (Line~\ref{line:expermental-design}).

Using the optimal design distribution \(\bm \lambda\upbra{\mathcal{S}}\), \est{} randomly samples \(N\) paths from the path set \(\mathcal{S}\) and obtains the path-level feedback \(Y_n\) for each sampled path \(k_n\) from the path-level \bench~(Lines~\ref{line:sample-loop}--\ref{line:sample-loop-end}).
These path-level observations are then used to set up a system of linear equations, one for every path \(k\in\mathcal{S}\), in \(\est\) (Lines~\ref{line:linear-regression-start}--Output), the solutions of which are estimates of the logarithmic depolarizing parameters $\{\log p_\ell\}_{\ell \in \mathcal L}$ (Lines~\ref{line:linear-regression-start}--Output).
The procedure is detailed in Algorithm~\ref{alg:link-estimate}. The following lemma depicts the final estimate's accuracy of \({\log \hat p_\ell}\).

\begin{lemma}[Accuracy of \est{}]
    \label{lem:link-estimate}
    Given path set \(\mathcal{K}\), link set \(\mathcal{L}\), a confidence parameter \(\delta\), and an error parameter \(\epsilon\),
    if {\est's} input sampling times
    \(N \ge \frac{C_0(4L + (6+\epsilon)L^2)}{\epsilon^2}\log \frac{5L}{\delta}\) where \(C_0\) is a constant depending on the \(C\) in Lemma~\ref{lma:benchmarking-concentration},
    then the \(\est(\mathcal{K}, N)\) algorithm (Algorithm~\ref{alg:link-estimate}) outputs \(\bm p_t\), and,  with probability at least \(1-\delta\), the estimates fulfill \[
        \mathbb{P}\left( \abs{\bm x^T(k) \left(\log \bm p - \log \hat{\bm p}_t\right)} \le \epsilon, \forall k\in\mathcal{K} \right) \ge 1 - \delta.
    \]
\end{lemma}

Lemma~\ref{lem:link-estimate} states that when the input sample size \(N\) is greater than the given threshold, \est{} accurately estimates the depolarizing parameters with a high probability such that the estimated fidelity \(\hat{F}(k) = \bm{x}^T(k) \log \hat{\bm{p}}_t\) is close to the true transformed fidelity \({F(k) = \bm{x}^T(k) \log \bm{p}}\) for all paths \(k \in \mathcal{K}\).


\subsubsection{\bpipath Algorithm}
As presented in Algorithm~\ref{alg:path-level}, \bpipath maintains a candidate arm set \(\mathcal{S}\), initialized as the full path set \(\mathcal{K}\), and iteratively prunes the candidate path set until only the best path remains. Specifically, the algorithm iterates over two nested loops. During each outer loop iteration indexed by subscript \(h\) (Line~\ref{line:half-iteration}), the algorithm halves the candidate path set.
Inside each outer iteration, the algorithm conducts an inner loop indexed by superscript \(s\) to successively prune the path set until the candidate path set is halved (Line~\ref{line:prune-iteration}).
Inside each inner iteration, the algorithm first sets the confidence parameter \(\delta\upbra{s}_h\) and the exploration parameter \(\xi\upbra{s}_h\) (Line~\ref{line:prune-iteration-parameters}) to estimate the depolarizing parameter of each link by calling the \est{} subroutine (Line~\ref{line:best-path}).
These detailed parameters are chosen according to Lemma~\ref{lem:link-estimate}.
The algorithm then prunes the path set by removing paths that are significantly worse than the best path (Line~\ref{line:prune-path}).


\begin{algorithm}[tb]
    \caption{\bpipath: {\small Path-Level Best Path Identification}}
    \label{alg:path-level}
    \begin{algorithmic}[1]
        \Input Path set \(\mathcal{K}\), link set \(\mathcal{L}\),
        confidence parameter \(\delta\)
        \State \textbf{initialization:} Candidate set \(\mathcal{S}_0 \gets  \mathcal{K}\)
        \For{\(h=0\) \textbf{to} \(\ceil{\log_2 L}\)}\label{line:half-iteration}
        \RightComment{halve path set in each \(h\) iteration}
        \State \(\mathcal{S}_h\upbra{1}\gets \mathcal{S}_h\) and \(s\gets 1\)
        \While{\(\abs*{\mathcal{S}_h\upbra{s}} > \floor{\frac{L}{2^h}}\)}\label{line:prune-iteration}
        \RightComment{prune path set in each \(s\) iteration}
        \State \(\delta\upbra{s}_h \gets \frac{36}{\pi^4}\cdot \frac{\delta}{(h+1)^2 s^2}\) and
        \(\epsilon_h\upbra{s} \gets \frac{1}{2^s}\)
        \label{line:prune-iteration-parameters}
        \State \resizebox{.815\hsize}{!}{\(
            \log \hat{\bm p}_h\upbra{s}
            \!\!\gets\!\! \est\!\!\left(\!\!\mathcal{S}_h\upbra{1}\!\!,  C_0\frac{2 + (6 + \epsilon_h\upbra{s}/4) L }{(\epsilon_h\upbra{s}/4)^2}\! \log \frac{5 \abs*{\mathcal{S}_h\upbra{1}}}{\delta_h\upbra{s}}\!\right)\)}

        \State \(k\upbra{s}_h \gets \oracle(\hat{\bm p}_h\upbra{s}, \mathcal{S}_h\upbra{s})\)
        \label{line:best-path}
        \State \resizebox{.815\hsize}{!}{\(\mathcal{S}_h\upbra{s+1}
        \!{\gets}\!
        \left\{\! k\!\in\!\mathcal{S}_h\upbra{s}\!\!:\!  ({\bm x}(k_h\upbra{s}) - {\bm x}(k))^T \! \log \hat{\bm p}_h\upbra{s}\!\! <\! \epsilon_h\upbra{s} \!\right\}\)}\label{line:prune-path}
        \RightComment{prune path set}

        \State \(s\gets s+1\)
        \EndWhile
        \State \(\mathcal{S}_{h+1} \gets \mathcal{S}_h\upbra{s}\) \Comment{halve path set}
        \EndFor
        \Output Path in \(\mathcal{S}_{h+1}\)
    \end{algorithmic}
\end{algorithm}

\subsection{Quantum Resource Complexity Analysis}
\label{subsec:path-analysis}

We present the resource complexity analysis of the \bpipath algorithm as follows,

\begin{theorem}[Resource complexity of \bpipath]\label{thm:path-complexity}
    Given confidence parameter \(\delta > 0\),    Algorithm~\ref{alg:path-level} outputs the best path with probability at least \(1-\delta\), and its total resources consumed is upper bounded as follows,
    \begin{equation}
        \label{eq:complexity-path-level}
        \complexity^{\pl} = O\left( \Lmax \sum_{\ell=2}^{L} \frac{1}{(\Delta^\pl([\ell]))^2} \log \frac{K}{\delta} \right),
    \end{equation}
    where path \([\ell]\) refers to the path with the \(\ell\)-th largest \(p^\pl(\cdot)\) among all paths, i.e., \(\Delta^\pl([\ell])\) is the \(\ell\)-th smallest path gap.
\end{theorem}

Applying the LinkSelFiE algorithm~\citep{liu2024link} to the \bpi model incurs \(
O(\Lmax \sum_{k=2}^K\frac{1}{(\Delta^\pl(k))^2} \log \frac K \delta  )
\) resource complexity,
where the summation range over all \(K\) paths can be much larger (\(O\left( L^{L_{\text{max}}} \right)\) in the worst case) than the range over the top \(L\) paths of \bpipath's complexity in~\eqref{eq:complexity-path-level}. Therefore, \bpipath can significantly reduce the resource complexity compared to LinkSelFiE.

Compared to \bpilink's complexity in~\eqref{eq:complexity-link-level}, the complexity of \bpipath in~\eqref{eq:complexity-path-level} has a similar form, but differs in two key aspects:
(1) the path-level complexity is based on the path gap \(\Delta^\pl([\ell])\) while the link-level complexity uses the link gap \(\Delta_{\ell}\), and both gaps are close in magnitude;
(2) the link-level complexity has a \(\Lmaxsquare\) factor, which is often worse than the linear \(\Lmax\) dependence of the path-level complexity.
However, in the worst case of \(K=O(L^{\Lmax})\),
the \bpipath's complexity would become  \(
O\left(  \Lmaxsquare\sum_{\ell=2}^{L} \frac{1}{(\Delta^\pl([\ell]))^2} \log \frac{L}{\delta} \right),
\)
which is on the same order as the \bpilink's complexity.
Therefore, both path-level and link-level benchmarking (and their corresponding algorithms) are of practical interest, depending on the specific quantum network scenario and the available quantum resources.

%% file: sections/qkd-ext.tex
\section{Learning Best Path for QKD}
\label{sec:qkd-extension}

In this section, we extend our results for identifying the path with the highest fidelity to the quantum key distribution (QKD) scenario, where we aim to find the path that maximizes QKD efficiency. We first introduce QKD (Section~\ref{subsec:qkd-background}) and define the secret key fraction (SKF) metric (Section~\ref{subsec:skf-utility}), which quantifies the efficiency of QKD. Then, we illustrate how to modify the \bpilink and \bpipath algorithms to identify the path with the largest SKF in Section~\ref{subsec:algorithm-qkd}.

\subsection{Background: Quantum Key Distribution and BB84 Protocol} \label{subsec:qkd-background}

Quantum key distribution (QKD) is an important application of quantum networks. It enables secure sharing of cryptographic keys between parties by leveraging quantum mechanics principles, such as the no-cloning theorem~\citep{wootters1982single}, to prevent interception by eavesdroppers. Multiple protocols have been devised for QKD, including BB84~\citep{bennet1984quantum}, E91~\citep{ekert1991quantum}, and B92~\citep{bennett1992quantum}, etc. Among them, the BB84 protocol, introduced by Bennett and Brassard in 1984~\citep{bennet1984quantum}, is the first and best known QKD protocol. The protocol proceeds as follows:

\emph{Stage 1: Quantum key initialization.}
Alice generates a random $N$-bit string \(\bm x_0\) and a random $N$-bit basis string.  She encodes \(\bm x_0\) into $N$ qubits using the basis string and transmits the qubits to Bob via a quantum channel. Bob also generates a random $N$-bit basis string and uses it to measure the received qubits. They publicly share their basis strings and retain only the bits where the bases match, forming an initial key \(\bm x_1\) of approximately \(N/2\) bits.
\emph{Stage 2: Noise and eavesdropper detection.}
Alice randomly selects a subset of \(\bm x_1\) for verification and sends the indices to Bob via a classical authentication channel. They compare these bits to determine the error rate \(\delta_{\text{error}}\). If the error rate exceeds a predefined threshold, they abort the protocol. The remaining bits form a key \(\bm x_2\) of about \(N/4\) bits.
\emph{Stage 3: Error correction and privacy amplification.}
Based on the error rate \(\delta_{\text{error}}\), Alice and Bob perform error correction and privacy amplification to derive a shorter, key \(\bm x_3\).

The efficiency of QKD is measured by the length of the final key \(\bm x_3\), which depends on the error rate influenced by quantum channel noise and potential eavesdropping. The fraction of the key retained after Stage 3 is referred to as the secret key fraction (SKF).
Ensuring a high SKF in a quantum network path enhances QKD security.

\subsection{QKD Objective: Secret Key Fraction}
\label{subsec:skf-utility}

To define the SKF, we first introduce the Werner state parameter \(w_\ell\) (resp., \(w^{\pl}(k)\)) for a quantum link \(\ell\) (resp., path \(k\)), which characterizes the noise level of the link (resp., path). This parameter depends on the error rate \(\delta_{\text{error}}\) in Stage 2 of BB84. Two properties of the Werner parameters are crucial:

\noindent
1. The Werner parameter of a path is the product of the Werner parameters of the links in the path~\citep{vardoyan2023quantum}, i.e.,
\begin{align}\label{eq:werner-state-parameter-product-relation}
    w^\pl(k) = \prod_{\ell \in \mathcal{L}(k)} w_\ell.
\end{align}
2. The Werner link parameter $w_\ell$ is related to the fidelity \(f_\ell\) and depolarizing parameter \(p_\ell\) as follows~\citep{vardoyan2023quantum}:
\begin{equation}\label{eq:werner-state-parameter-depolarizing-relation}
    w_\ell \overset{(a)}= \frac{4f_\ell - 1}{3} \overset{(b)}= \frac{2p_\ell + 1}{3},
\end{equation}
where equation (a) of~\eqref{eq:werner-state-parameter-depolarizing-relation} follows from the relationship between the Werner state parameter and the fidelity of the generated entanglement pair across the link \(f_\ell\)~\citep{vardoyan2023quantum}, i.e., \(f_\ell = \frac{3w_\ell + 1}{4}\), and equation (b) of~\eqref{eq:werner-state-parameter-depolarizing-relation} follows from the relationship between the link fidelity and the depolarizing parameter \(p_\ell\)~\citep{helsen2023benchmarking}, i.e., \(f_\ell = \frac{p_\ell + 1}{2}\).\footnote{Note that the fidelity definition in~\citep{helsen2023benchmarking} refers to the average fidelity for a quantum channel, which differs from the entanglement fidelity of the generated entanglement pair. However, as illustrated by~\citet{helsen2023benchmarking}, when the quantum channel is implemented via teleportation, and the entanglement pair is generated by the same, the fidelities are closely related.}

Denote \(u(k)\) as the SKF~\citep{shor2000simple} of deploying BB84~\citep{bennet1984quantum} on path \(k\). This SKF can be expressed as:
\[
    u(k) \coloneqq 1 - 2h\left( \frac{1-w^\pl(k)}{2} \right),
\]
where
\(h(x) = -x\log_2 x - (1-x) \log_2 (1-x)\) is the binary Shannon entropy. The SKF metric can be extended to the secret key rate (SKR) is the product of \(u(k)\) and the quantum channel capacity~\citep{shor2000simple}. For simplicity, we focus on SKF in this paper.
The SKF function \(u(k)\) of path \(k\) can be transformed by a monotonic bijection as follows,
\begin{align}
    u(k) & \longleftrightarrow  - h\left( \frac{1-w^\pl(k)}{2} \right) \nonumber                                                                 \\
         & \overset{(a)}\longleftrightarrow  w^\pl(k) \overset{(b)}\longleftrightarrow \sum_{\ell \in \mathcal{L}(k)} \log w_\ell      \nonumber \\
         & = \sum_{\ell \in \mathcal{L}(k)} \log \left(  \frac{2p_\ell + 1}{3} \right) \label{eq:fidelity-to-depolarizing}
    \eqqcolon U(k),
\end{align}
where mapping (a) follows from the monotonic increase of the binary entropy function \(h(x)\) for \(0 \le x \le \frac{1}{2}\), mapping (b) follows from~\eqref{eq:werner-state-parameter-product-relation} and the monotonic increase of the logarithm function, and \eqref{eq:fidelity-to-depolarizing} follows from \eqref{eq:werner-state-parameter-depolarizing-relation}.
As the transformed SKF \(U(k)\) has a monotonic bijection with SKF \(u(k)\), our objective becomes that of finding the path \(k^*\) that maximizes \(U(k)\).

\subsection{Algorithms for Learning the Best Path for QKD}\label{subsec:link-qkd}\label{subsec:algorithm-qkd}

Compared to the transformed fidelity \(F(k) = \sum_{\ell\in \mathcal{L}(k)} \log p_\ell\),
the transformed SKF function \(U(k)\) in \eqref{eq:fidelity-to-depolarizing} involves an additional linear transformation of the depolarizing parameters, namely replacing \(p_\ell\) with \({(2p_\ell + 1)}/{3}\).
Due to this transformation, the best path in terms of fidelity may not be the best path in terms of SKF.
To identify the best path for SKF, one needs to modify \bpilink and \bpipath to identify the best path for QKD with link-level and path-level benchmarking, respectively.

For link-level benchmarking, the only \oracle{} requires modification. Instead of finding the empirical shortest (best) path \(\hat k_t\) using edge lengths \(-\log p_\ell\), the \oracle{} now uses edge lengths \(-\log \left({(2 p_\ell + 1)}/{3}\right)\). The estimate transformation of the modified \bpilink algorithm is from the link depolarizing parameter estimate \(\hat p_\ell\)
to \(\log ((2\hat p_\ell + 1)/3)\).





For path-level benchmarking, modifying \bpipath is more involved. Recall that the linear regression \est{} subroutine (Algorithm~\ref{alg:link-estimate}) estimates the depolarizing parameter \(\log \hat p_\ell\) for each link \(\ell\). However, in the QKD scenario, we need to estimate \(\log ((2\hat p_\ell + 1)/3))\) for the SKF \(U(k)\) in~\eqref{eq:fidelity-to-depolarizing}.
Therefore, \est{} for maximizing SKF needs to go through the following transformation,
\(\hat p^\pl(k) \overset{(a)}\to \log \hat p_\ell \overset{(b)}\to \hat p_\ell \overset{(c)}\to \log ((2\hat p_\ell + 1)/3).\)
Transformation (a) transforms path-level benchmarking into link-level benchmarking by \(\est\). Transformation (b) applies an exponential function to negate the logarithm, and (c) converts the depolarizing parameter into the Werner parameter.
The following lemma depicts the accuracy of the final estimate of \({\log ((2\hat p_\ell + 1)/3)}\). It shows that the estimation errors of the modified \est{} are twice as large as that for estimating \(\log \hat p_\ell\), requiring an adjustment in the path prune condition in Line~\ref{line:prune-path} of \bpipath.

\begin{lemma}[Accuracy of \est{} for QKD]
    \label{lem:link-estimate-qkd}
    Under the same conditions as in Lemma~\ref{lem:link-estimate}, the estimate guarantees that  \[
        \mathbb{P}\left( \abs{\bm x^T(k)\! \left(\!\log \frac{2\bm p + 1}{3} - \log \frac{2\hat{\bm p}_t + 1}{3}\!\right)\!} \!\le\! 2 \epsilon, \forall k\in\mathcal{K} \!\right) \ge 1 - \delta.
    \]
\end{lemma}

In summary, when SKF is the metric, \bpipath requires two modifications:
(1) add the processes (a), (b), and (c) to \est{} to estimate \(\log ((2\hat p_\ell + 1) / 3)\), and (2) revise the path prune condition in Line~\ref{line:prune-path} from \(\epsilon\) to \(2\epsilon\).

The modified \bpilink and \bpipath algorithms can identify the path with the maximum SKF in the quantum network using link-level and path-level benchmarking, respectively. The corresponding complexity analysis results remain the same as in Theorems~\ref{thm:link-complexity} and~\ref{thm:path-complexity} (in the big-\(O\) perspective) with the gaps \(\Delta_\ell\) and \(\Delta^{\pl}(k)\) defined with respect to the SKF \(U(k)\), and thus the details are omitted here.

%% file: sections/simu.tex
\section{Simulation}\label{sec:simulation}

In this section, the performance of \bpilink and \bpipath are reported on both high-fidelity and high-SKF path identification tasks.
We first describe our experimental setups and present our results, focusing on resource complexity compared to the baselines.

\subsection{Experiment Setup}
\noindent \textbf{Network Topologies.} To demonstrate the advantage of our proposed algorithms,
we use a simple yet typical topology where the number of paths is around \(4\) times larger than the number of links,
as shown in Figure~\ref{fig:topology}.
The advantage of our algorithms would be more obvious for topologies where the number of paths significantly exceeds the number of links.
By varying the number of parallel links between adjacent nodes, we construct network topologies with different numbers of paths.
Assigned fidelities ensure a unique best path.

\noindent \textbf{Baseline Algorithms.}
We compare our algorithms to the following baseline algorithms: (1) LinkSelFiE~\cite{liu2024link}, (2) SuccElim~\citep{even2006action}, (3) Uniform-Path, and (4) Uniform-Link.
LinkSelFiE, SuccElim, and Uniform-Path are path-level algorithms that apply \(\bench\) to quantum paths.
LinkSelFiE and SuccElim use elimination strategies to identify the best path.
Uniform-Path (resp., Uniform-Link) uniformly applies \(\bench\) over all paths (resp., links) and chooses the empirical best path.
$T_0$ is set to 200 for Uniform-Path and Uniform-Link to ensure accurate estimations.
A smaller value of \(T_0\), 10 is used for SuccElim, \bpilink, and \bpipath, as these algorithms iteratively benchmark links or paths multiple times.
We specify the bounce number set \(\mathcal{M} = \{1,2,\dots, 10\}\) for all algorithms, following the guidelines in~\citep{helsen2023benchmarking}.



\noindent \textbf{Noise Models.} We simulate quantum noise using four typical quantum noise models~\citep{nielsen2001quantum}: (1) depolarizing, (2) dephasing, (3) amplitude damping, and (4) bit flip noise. These noise models are common in NISQ era quantum networks.
We define the fidelity of each link in Figure~\ref{fig:topology} as follows:
For links between nodes C and D, one link has a fidelity of 0.99, and the fidelities of other links decrease from 0.95, reducing by 0.1 per link.
For the two links between nodes A and B, and nodes B and C, the fidelities are set to 0.99 and 0.90, respectively.
This setup ensures that the best path is unique and that the paths are distinguishable.
To ensure a fair comparison between these noise models, we convert a given fidelity value into the corresponding noise parameters required to initialize each model.
All network benchmarking operations, e.g., the Clifford gates, are assumed to be noiseless. If they are noisy, one can first benchmark these gates individually~\cite {knill2008randomized} and then compensate for their impact on the network benchmarking.

All algorithms are simulated by an off-the-shelf quantum network using NetSquid~\cite{coopmans2021netsquid}.
We average the evaluation results over 10 trials and include the resulting error bars in each figure.

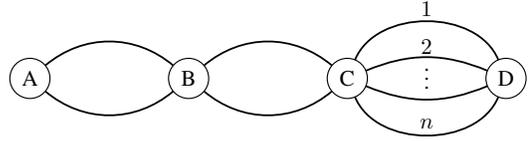
\begin{figure}
  \centering
  \resizebox{0.8\linewidth}{!}{
    \begin{tikzpicture}
      \SetGraphUnit{2.5} 
      \Vertices{line}{A,B,C,D} 
      \node at ($(C)!.5!(D)+(0,0.1)$) {$\vdots$};
      \node at ($(C)!.5!(D)+(0,1.1)$) {$1$};
      \node at ($(C)!.5!(D)+(0,0.5)$) {$2$};
      \node at ($(C)!.5!(D)+(0,-0.72)$) {$n$};
      \foreach \from/\to in {A/B,B/C}
        {
          \foreach \i in {1,...,2}
            {
              \pgfmathsetmacro{\bend}{80*(\i-1.5)}
              \Edge[style={bend left=\bend}](\from)(\to) 
            }
        }
      \foreach \from/\to in {C/D}
        {
          \foreach \i in {1,...,4}
            {
              \pgfmathsetmacro{\bend}{45*(\i-2.5)}
              \Edge[style={bend left=\bend}](\from)(\to) 
            }
        }
    \end{tikzpicture}
  }
  \caption{An example of network topologies used in our experiments. It has $n+4$ links and $4n$ paths between A and D.}
  \label{fig:topology}
\end{figure}

\begin{figure}[tbp]
  \centering
  \includegraphics[width=\linewidth]{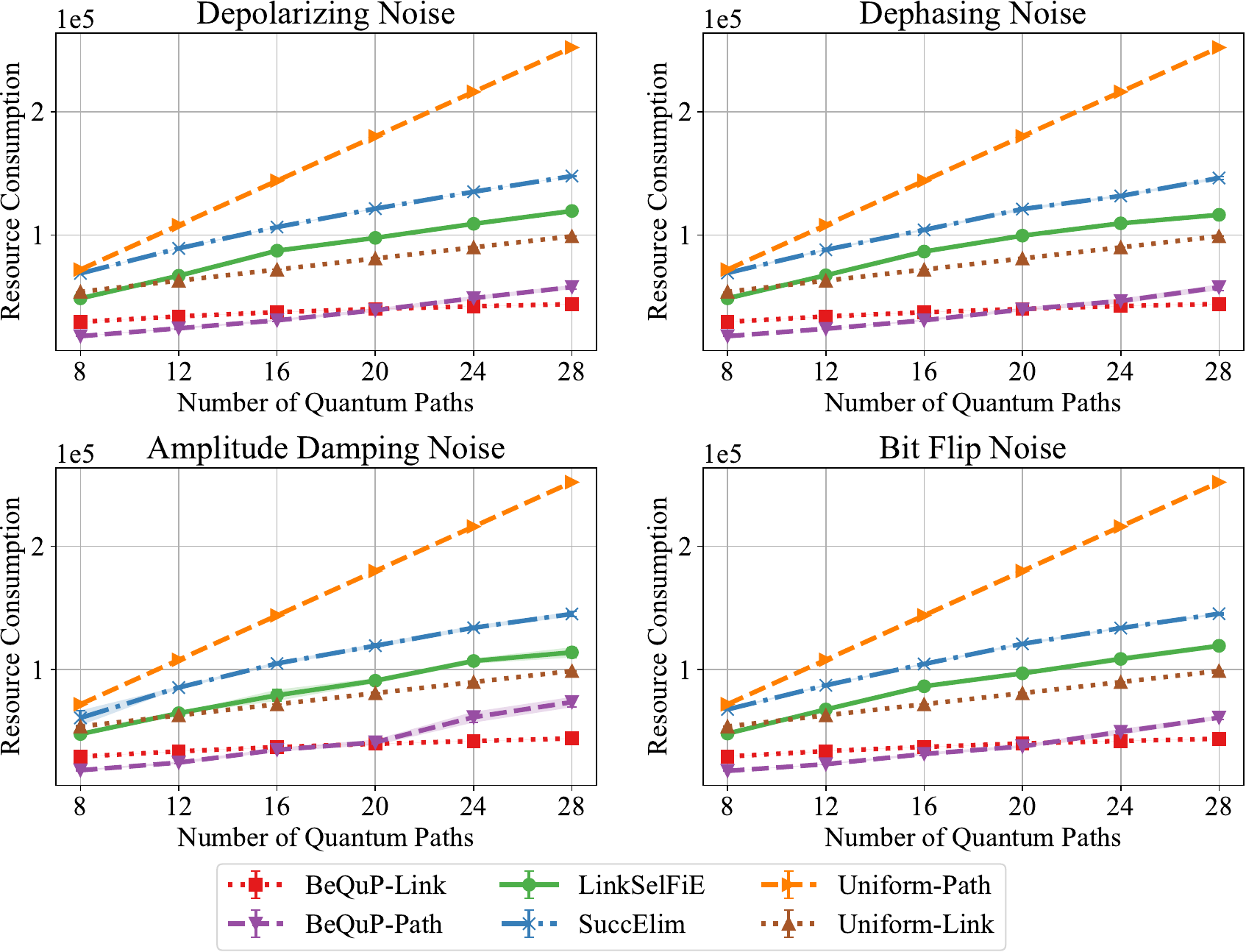}
  \caption{Comparison of quantum resource complexity under different noise models for high-fidelity path identification.}
  \label{fig:cost-and-path-num}
\end{figure}
\begin{figure}[htbp]
  \centering
  \includegraphics[width=\linewidth]{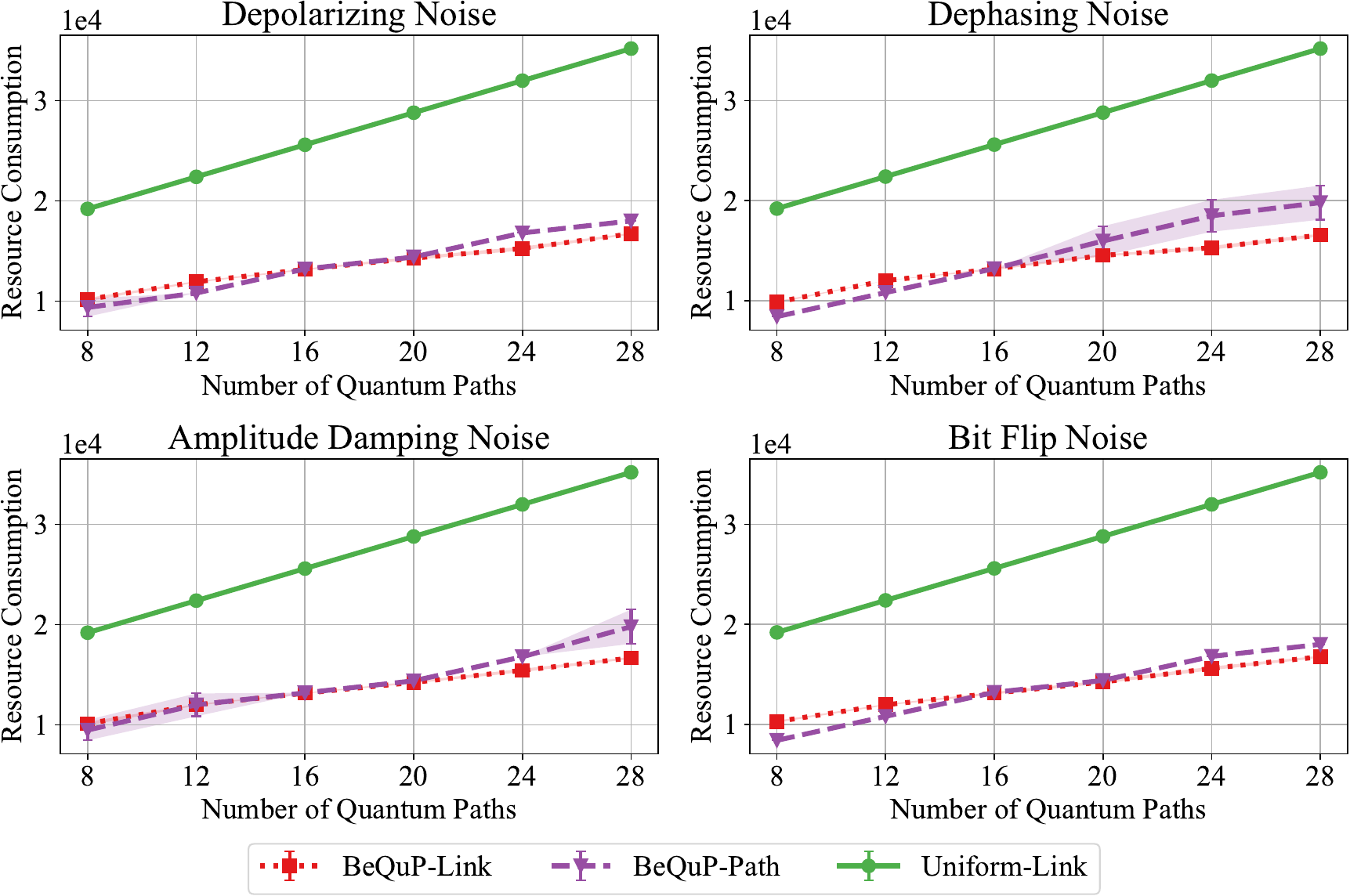}
  \caption{Comparison of quantum resource complexity under different noise models for high-SKF path identification.}
  \label{fig:skf-cost-and-path-num}
\end{figure}

\subsection{High-fidelity Path Identification}
We first evaluate the performance of high-fidelity path identification algorithms by applying them to networks with different numbers of paths $K=4n$ by varying \(n\in\{2,3,\dots, 7\}\) in Figure~\ref{fig:topology}
and measuring resource consumption (complexity).
Upon termination, all algorithms successfully identify the best path.
Figure~\ref{fig:cost-and-path-num} shows the relationship between resource consumption and the number of paths under different noise models, where resource consumption is quantified by the total number of quantum entangled pairs consumed.
As shown in Figure~\ref{fig:cost-and-path-num},
Uniform-Path's resource complexity scales in proportion to the number of paths because it allocates equal resources to each path.
Likewise, the resource complexity for Uniform-Link grows with the number of paths due to an increase in the number of involved links based on the network topology.
In contrast, LinkSelFiE, SuccElim, \bpilink, and \bpipath can adaptively allocate resources based on the quality of links or paths.
However, LinkSelFiE and SuccElim are topology oblivious, meaning that they treat each path as independent and overlook overlapping structures, resulting in significantly higher resource complexity.
Conversely, \bpipath leverages the overlapping links among paths to optimize resource complexity, and \bpilink keeps track of each quantum link's estimates and selectively benchmarks links to distinguish between paths, thereby achieving the best performance.
The curve for \bpipath has a higher slope than that of \bpilink with respect to the number of paths. Because the additional \est{} subroutine of \bpipath, for transferring the path-level observations to the link-level parameter, introduces a higher estimation error, resulting in a higher constant coefficient in its total quantum resource complexity.

\subsection{Best SKF Path Identification}
We assess the algorithms performance for identifying the best SKF paths by applying them across networks with varying numbers of paths $K$ and measuring the quantum resource complexity.
Obtaining the SKF for a path requires information about all links in the path.
However, LinkSelFiE, SuccElim, and Uniform-Path have no link estimation process, making them infeasible for this task.
Hence, we compare our algorithms \bpilink and \bpipath against Uniform-Link only.
The results are shown in Figure~\ref{fig:skf-cost-and-path-num}.
Similar to the trends observed in high-fidelity path identification,
both \bpilink and \bpipath achieve substantially lower resource complexity than Uniform-Link across all noise models, demonstrating their efficiency.

%% file: sections/relat.tex
\section{Related Works}\label{sec:related-works}

Path selection in quantum networks is crucial for achieving long-distance quantum communication~\citep{van2013path,liu2024link,liu2024quantum,zhao2022e2e,zeng2022multi,liu2022quantum}. Most prior works in this domain have focused on formulating quantum network models~\citep{van2013path,vardoyan2023quantum}, proposing protocols for path selection (routing)~\citep{zeng2022multi},
and evaluating the performance of quantum communication protocols~\citep{van2013path}.

\citet{van2013path} were among the first to formalize the path selection problem in quantum networks, using Bell pairs per second as their metric and evaluating the performance of Dijkstra's algorithm. Although this metric is useful for assessing general quantum network performance, it is not directly applicable to the channel fidelity and the SKF metrics, which are our main focus.
\citet{vardoyan2023quantum} examined the quantum network utility maximization problem by formulating three different utility functions: distillable entanglement, secret key fraction (SKF), and negativity. While they did consider the SKF metric, their work centered on utility maximization as an offline optimization task and did not address the path selection problem or the design of online algorithms, which is our focus.

Recently, \citet{liu2024link} studied the link selection problem in quantum networks, simplifying paths between two nodes to a single link and proposing an online algorithm, LinkSelFiE, for selecting the link with the highest fidelity. While LinkSelFiE can be generalized to the path selection problem, it suffers from the curse of dimensionality when the number of paths is large, as detailed in Section~\ref{subsec:path-analysis}. Our work addresses this issue by introducing a novel online algorithm for the path selection problem in quantum networks, applicable to both the fidelity and the SKF metrics, even in scenarios with many paths.

Learning the best action among a set of possible actions is a fundamental problem in reinforcement learning~\citep{even2006action,chen2014combinatorial,tao2018best}. \citet{even2006action} investigate the identification of the optimal action from multiple individual actions within the stochastic bandit setting, where each arm corresponds to a stochastic distribution of rewards. \citet{chen2014combinatorial,du2021combinatorial} focus on the combinatorial pure exploration (CPE) problem, which involves finding the best combination of individual actions that yields the highest reward among a set of action combinations. Specifically, \citet{chen2014combinatorial} explore the additive reward function, and \citet{du2021combinatorial} examine the bottleneck function.


\citet{tao2018best} study best arm identification in linear bandits, aiming to determine the arm with the highest reward among a set of arms with linear reward functions. Their focus on the linear reward function structure relates to the multiplicative utility function in our \bpi problem with path-level feedback. However, their research emphasizes the linear additive reward function, whereas the multiplication function of depolarizing parameters or the additive function for the logarithm of rewards, as described in~\eqref{eq:fidelity-to-depolarizing}, has not yet been explored. Our \bpipath algorithm introduces a new parameter transformation procedure to estimate the depolarizing parameters of paths based on the path-level feedback in Section~\ref{subsec:path-algorithm}.

%% file: sections/concl.tex
\section{Conclusion}\label{sec:conclusion}



This paper addressed the challenge of identifying the best paths, in terms of both fidelity and SKF, in a quantum network within an online learning framework. We introduced two online learning algorithms, \bpilink and \bpipath, designed to determine the best path by utilizing link-level and path-level benchmarking, respectively. Theoretical guarantees were provided for both algorithms, demonstrating their efficiency in minimizing quantum resource costs. Additionally, we validated our algorithms using NetSquid simulations, showing their performance against established baselines.

%% file: sections/proof-link.tex
\subsection{Proof for Algorithm~\ref{alg:link-level}}

We first define event \(\mathcal{E}_t\) as \(
\mathcal{E}_t \coloneqq \left\{ \forall \ell \in \mathcal{L}: \abs{\hat p_{\ell,t} - p_\ell} \le \rad_{\ell, t}\right\}.
\)
In the following, we prove the theorem in three steps:
(1) prove that the event \(\cap_{t=1}^\infty \mathcal{E}_t\) holds with probability at least \(1 - \delta\);
(2) \emph{(correctness)} given that the event \(\cap_{t=1}^\infty \mathcal{E}_t\) holds, prove that the algorithm identifies the best path;
(3) \emph{(complexity)} given that the event \(\cap_{t=1}^\infty \mathcal{E}_t\) holds, prove that the cost complexity is \(O\left( \Lmax^2 H \log \left( \frac{LH}{\delta} \right) \right)\).

\subsubsection{High Probability Event}

By letting \(\delta \gets \frac{\delta}{2Lt^3}\) in Lemma~\ref{lma:benchmarking-concentration}, we have that, given \(\hat p_{\ell,t}\) as the average of \(N\) samples, \begin{equation}\label{eq:concentration-link-level}
    \mathbb{P}\left( \abs{\hat p_{\ell,t}(N)  - p_\ell}
    \le \sqrt{\frac{C\log \frac{2Lt^3}{\delta}}{N}}\right)
    \ge 1 - \frac{\delta}{2Lt^3}.
\end{equation}

As \(\rad_{\ell, t} = \sqrt{\frac{C\log\left( \frac{2Lt^3}{\delta} \right)}{N_{\ell,t}}}\), we have \begin{equation}
    \label{eq:single-event}
    \begin{split}
        \mathbb{P}(\mathcal{E}_t)
         & = 1 - \mathbb{P}\left( \exists \ell \in \mathcal{L}: \abs{\hat p_{\ell,t} - p_\ell} > \rad_{\ell, t}\right)
        \\
         & \ge 1 - \sum_{\ell \in \mathcal{L}} \mathbb{P}\left( \abs{\hat p_{\ell,t}(N_{\ell,t})  - p_\ell} > \sqrt{\frac{C\log \frac{2Lt^3}{\delta}}{N_{\ell,t}}} \right)
        \\
         & \overset{(a)}\ge 1 - \sum_{\ell \in \mathcal{L}} \sum_{N=1}^t \mathbb{P}\left( \abs{\hat p_{\ell,t}(N)  - p_\ell} > \sqrt{\frac{C\log \frac{2Lt^3}{\delta}}{N}} \right)
        \\
         & \overset{(b)}\ge 1 - \sum_{\ell \in \mathcal{L}} \sum_{N=1}^t  \frac{\delta}{2Lt^3} = 1 - \frac{\delta}{2t^2},
    \end{split}
\end{equation}
where inequality (a) is due to the fact that the number of samples \(N_{\ell,t}\) is at most \(t\), and inequality (b) is due to the concentration inequality in~\eqref{eq:concentration-link-level}.

Therefore, we have \[
    \mathbb{P}\left( \cap_{t=1}^\infty \mathcal{E}_t \right)
    \ge 1 - \sum_{t=1}^\infty \mathbb{P}(\neg \mathcal{E}_t) \overset{(a)}\ge 1 - \sum_{t=1}^\infty \frac{\delta}{2t^2} = 1 - \frac{\pi^2}{12}\delta \ge 1 - \delta,
\]
where inequality (a) is due to~\eqref{eq:single-event}.

\subsubsection{Correctness}

In this subsubsection, we show that when the event \(\cap_{t=1}^\infty \mathcal{E}_t\) holds, the algorithm identifies the best path. That is, when the algorithm terminates at time slot \(t\), the output arm \(\hat k_t\) is equal to the best path \(k^*\). Following the notations of the depolarizing parameter \(p_\ell\), the empirical mean \(\hat{p}_{\ell,t}\), and the confidence bound \(\tilde{p}_{\ell,t}\) in Algorithm~\ref{alg:link-level}, we define \(\hat{w}_{\ell,t} \coloneqq \frac{2\hat{p}_{\ell,t} + 1}{3} , \tilde{w}_{\ell,t} \coloneqq  \frac{2\tilde{p}_{\ell,t} + 1}{3}\)
and \(\hat u_t(k) \coloneqq \sum_{\ell \in \mathcal{L}(k)}\log \hat w_{\ell,t}, \tilde u_t(k) \coloneqq \sum_{\ell \in \mathcal{L}(k)}\log \tilde w_{\ell,t}\).

If the terminated output \(\hat k_t \neq k^*\), then there exists another \(k'_t\neq \hat k_t\) such that \(u(k'_t) > u(\hat k_t)\). Recalling that \(\tilde k_t =\argmax_{k\in\mathcal{K}} \tilde u(k)\), we also have \(\tilde u_t(\tilde k_t) \ge  \tilde u_t(k'_t)\).

\begin{align}
    \nonumber
    \tilde u_t(\tilde k_t) - \tilde u_t (\hat k_t)
     & \ge \tilde u_t(k'_t) - \tilde u_t (\hat k_t)
    \\\nonumber
     & = \sum_{\ell \in \mathcal{L}(k'_t)} \log\tilde w_{\ell,t} - \sum_{\ell \in \mathcal{L}(\hat k_t)} \log\tilde w_{\ell,t}
    \\\nonumber
     & = \sum_{\ell \in \mathcal{L}(k'_t) \setminus \mathcal{L}(\hat k_t)} \log\tilde w_{\ell,t} - \sum_{\ell \in \mathcal{L}(\hat k_t) \setminus \mathcal{L}(k'_t)} \log\tilde w_{\ell,t}
    \\\nonumber
     & = \sum_{\ell \in \mathcal{L}(k'_t) \setminus \mathcal{L}(\hat k_t)} \log \left( \frac{2(\hat p_{\ell,t} + \rad_{\ell,t}) + 1}{3} \right) - \sum_{\ell \in \mathcal{L}(\hat k_t) \setminus \mathcal{L}(k'_t)} \log \left( \frac{2(\hat p_{\ell,t} - \rad_{\ell,t}) + 1}{3} \right)
    \\\label{eq:bound-to-true-mean}
     & \ge \sum_{\ell \in \mathcal{L}(k'_t) \setminus \mathcal{L}(\hat k_t)} \log \left( \frac{2p_{\ell} + 1}{3} \right) - \sum_{\ell \in \mathcal{L}(\hat k_t) \setminus \mathcal{L}(k'_t)} \log \left( \frac{2 p_{\ell} + 1}{3} \right)
    \\\nonumber
     & = u(k'_t) - u(\hat k_t) > 0,
\end{align}
where~\eqref{eq:bound-to-true-mean} is due to the fact that the event \(\cap_{t=1}^\infty \mathcal{E}_t\) holds, i.e., \(p_\ell \in \left( \hat p_{\ell,t} - \rad_{\ell,t}, \hat p_{\ell, t} + \rad_{\ell,t} \right)\).

The above inequality shows that \(\tilde u_t(\tilde k_t) > \tilde u_t (\hat k_t)\), which contradicts the fact the stopping condition \(\tilde u_t(\tilde k_t) = \tilde u_t (\hat k_t)\).
Therefore, we prove the terminated output \(\hat k_t = k^*\).

\subsubsection{Complexity}

To show the complexity of Algorithm~\ref{alg:link-level}, we prove that for any time slot \(t\) and link \(\ell'\), assuming event \(\mathcal{E}_t\) holds, if \(\rad_{\ell',t} > \frac{\Delta_{\ell'}}{6 \Lmax}\), then link \(\ell'\) is not selected at this time slot, that is, \(\ell_t \neq \ell'\).

To prove the above statement via contradiction, we assume that \(\ell_t = \ell'\) and \(\rad_{\ell',t} \le \frac{\Delta_{\ell'}}{6 \Lmax}\). There are two cases of the link \(\ell'\): \begin{itemize}
    \item Case (1): \(\ell' \in ((\mathcal{L}(\tilde k_t) \setminus \mathcal{L}(\hat k_t))\cap \mathcal{L}(k^*)) \cup ((\mathcal{L}(\hat k_t) \setminus \mathcal{L}(\tilde k_t))\setminus \mathcal{L}(k^*))\), implying \(\hat k_t \neq k^*\);
    \item Case (2): \(\ell' \in ((\mathcal{L}(\tilde k_t) \setminus \mathcal{L}(\hat k_t))\setminus \mathcal{L}(k^*)) \cup ((\mathcal{L}(\hat k_t) \setminus \mathcal{L}(\tilde k_t))\cap \mathcal{L}(k^*))\), implying \(\tilde k_t \neq k^*\).
\end{itemize}

For Case (1), we have \begin{align}
    \nonumber
    \hat u_t(\tilde k_t) - \hat u_t(\hat k_t)
     & = \sum_{\ell \in \mathcal{L}(\tilde k_t)} \log\hat w_{\ell,t} - \sum_{\ell \in \mathcal{L}(\hat k_t)} \log\hat w_{\ell,t}
    \\\nonumber
     & = \sum_{\ell \in \mathcal{L}(\tilde k_t) \setminus \mathcal{L}(\hat k_t)} \log\hat w_{\ell,t} - \sum_{\ell \in \mathcal{L}(\hat k_t) \setminus \mathcal{L}(\tilde k_t)} \log\hat w_{\ell,t}
    \\\label{eq:case-1-scale-hat-w-to-tilde-w}
     & \ge \sum_{\ell \in \mathcal{L}(\tilde k_t) \setminus \mathcal{L}(\hat k_t)} \left(\log \tilde w_{\ell, t} - \frac{2}{3} \rad_{\ell,t} \right) - \sum_{\ell \in \mathcal{L}(\hat k_t) \setminus \mathcal{L}(\tilde k_t)} \left( \log\tilde w_{\ell, t} + \frac{2}{3} \rad_{\ell,t} \right)
    \\\nonumber
     & = \tilde u_t(\tilde k_t) - \tilde u_t(\hat k_t) - \frac{2}{3}  \sum_{\ell \in \mathcal{L}(\tilde k_t) \bigtriangleup \mathcal{L}(\hat k_t)} \rad_{\ell,t}
    \\\label{eq:case-1-tilde-k-to-k-star}
     & \ge \tilde u_t(k^*) - \tilde u_t(\hat k_t) - \frac{2}{3}  \sum_{\ell \in \mathcal{L}(\tilde k_t) \bigtriangleup \mathcal{L}(\hat k_t)} \rad_{\ell,t}
    \\\nonumber
     & = \sum_{\ell \in \mathcal{L}(k^*) \setminus \mathcal{L}(\hat k_t)} \log\tilde w_{\ell,t} - \sum_{\ell \in \mathcal{L}(\hat k_t) \setminus \mathcal{L}(k^*)} \log\tilde w_{\ell,t}
    - \frac{2}{3}  \sum_{\ell \in \mathcal{L}(\tilde k_t) \bigtriangleup \mathcal{L}(\hat k_t)} \rad_{\ell,t}
    \\\label{eq:case-1-tilde-w-to-w}
     & \ge \sum_{\ell \in \mathcal{L}(k^*) \setminus \mathcal{L}(\hat k_t)}  \log w_{\ell} - \sum_{\ell \in \mathcal{L}(\hat k_t) \setminus \mathcal{L}(k^*)}  \log w_{\ell}
    - \frac{2}{3}  \sum_{\ell \in \mathcal{L}(\tilde k_t) \bigtriangleup \mathcal{L}(\hat k_t)} \rad_{\ell,t}
    \\\nonumber
     & = u(k^*) - u(\hat k_t) - \frac{2}{3}  \sum_{\ell \in \mathcal{L}(\tilde k_t) \bigtriangleup \mathcal{L}(\hat k_t)} \rad_{\ell,t}
    \\\label{eq:case-1-to-gap}
     & \ge \Delta_{\ell'} - \frac{2\Delta_{\ell'}}{9}
    \\\nonumber
     & > 0,
\end{align}
where~\eqref{eq:case-1-scale-hat-w-to-tilde-w} is because \(\log \hat w_{\ell,t} = \log\frac{2\hat p_{\ell, t} + 1}{3} =
\log \left( \frac{2 (\hat p_{\ell,t} + \rad_{\ell,t}) + 1}{3} - \frac{2\hat p_{\ell,t}}{3}\right)
\ge \log \frac{2 (\hat p_{\ell,t} + \rad_{\ell,t}) + 1}{3} - \frac{2\hat p_{\ell,t}}{3} = \log\tilde w_{\ell,t} - \frac{2\hat p_{\ell,t}}{3}\) for \(\ell \not\in \mathcal{L}(\hat k_t)\)
and \(\log\hat w_{\ell,t} \le \tilde \log w_{\ell,t} + \frac{2}{3} \rad_{\ell,t}\) for \(\ell \in \mathcal{L}(\hat k_t)\),
\eqref{eq:case-1-tilde-k-to-k-star} is due to that \(\tilde k_t \in \argmax_{k\in\mathcal{K}} \tilde u_t(k)\),
\eqref{eq:case-1-tilde-w-to-w} is because \(\tilde w_{\ell,t} \ge  w_\ell\) for \(\ell \not \in \mathcal{L}(\hat k_t)\)
and \(\tilde w_{\ell,t} \le w_\ell\) for \(\ell \in \mathcal{L}(\hat k_t)\),
and~\eqref{eq:case-1-to-gap} is by the definition of \(\Delta_{\ell'}\) and the condition \(\rad_{\ell',t} \le \frac{\Delta_{\ell'}}{6 \Lmax}\).
The inequality \(\hat u_t(\tilde k_t) > \hat u_t (\hat k_t)\) contradicts the definition of \(\hat k_t\).

For Case (2), we have
\begin{align}
    \nonumber
    \tilde u_t (k^*) - \tilde u_t (\tilde k_t)
     & =
    \tilde u_t (k^*) - \tilde u_t (\hat k_t) + \tilde u_t (\hat k_t) - \tilde u_t (\tilde k_t)
    \\\nonumber
     & = \sum_{\ell \in \mathcal{L}(k^*)} \log\tilde w_{\ell,t} - \sum_{\ell \in \mathcal{L}(\hat k_t)} \log\tilde w_{\ell,t}
    + \sum_{\ell \in \mathcal{L}(\hat k_t)} \log\tilde w_{\ell,t} - \sum_{\ell \in \mathcal{L}(\tilde k_t)} \log\tilde w_{\ell,t}
    \\\nonumber
     & = \sum_{\ell \in \mathcal{L}(k^*) \setminus \mathcal{L}(\hat k_t)} \log\tilde w_{\ell,t}
    - \sum_{\ell \in \mathcal{L}(\hat k_t) \setminus \mathcal{L}(k^*)} \log\tilde w_{\ell,t}
    + \sum_{\ell \in \mathcal{L}(\hat k_t) \setminus \mathcal{L}(\tilde k_t)} \log\tilde w_{\ell,t} - \sum_{\ell \in \mathcal{L}(\tilde k_t) \setminus \mathcal{L}(\hat k_t)} \log\tilde w_{\ell,t}
    \\\label{eq:case-2-tilde-w-to-w}
     & \ge \sum_{\ell \in \mathcal{L}(k^*) \setminus \mathcal{L}(\hat k_t)} \log w_{\ell}
    - \sum_{\ell \in \mathcal{L}(\hat k_t) \setminus \mathcal{L}(k^*)} \log w_{\ell}
    \\\label{eq:case-2-tilde-w-to-w-2}
     & \qquad\qquad\qquad\qquad
    + \sum_{\ell \in \mathcal{L}(\hat k_t) \setminus \mathcal{L}(\tilde k_t)} \left( \log w_{\ell} - \frac 4 3 \rad_{\ell,t} \right)
    - \sum_{\ell \in \mathcal{L}(\tilde k_t) \setminus \mathcal{L}(\hat k_t)} \left(\log w_{\ell} - \frac 4 3 \rad_{\ell,t} \right)
    \\\nonumber
     & = u(k^*) - u(\tilde k_t) - \frac 4 3  \sum_{\ell \in \mathcal{L}(\tilde k_t) \bigtriangleup \mathcal{L}(\hat k_t)} \rad_{\ell,t}
    \\\label{eq:case-2-to-gap}
     & \ge \Delta_{\ell'} - \frac{4\Delta_{\ell'}}{9}
    \\\nonumber
     & > 0,
\end{align}
where~\eqref{eq:case-2-tilde-w-to-w}-\eqref{eq:case-2-tilde-w-to-w-2} are because \(\log \tilde w_{\ell,t} = \log \left( \frac{2(\hat p_{\ell,t} - \rad_{\ell,t}) + 1}{3} \right) = \log \left( \frac{2(\hat p_{\ell,t} + \rad_{\ell,t}) + 1}{3} - \frac{4\rad_{\ell,t}}{3} \right)
\ge \log \frac{2(\hat p_{\ell,t} + \rad_{\ell,t}) + 1}{3} - \frac{4\rad_{\ell,t}}{3} \ge  \log w_\ell - \frac{4\rad_{\ell,t}}{3}\) for \(\ell \in \mathcal{L}(\hat k_t)\) and \(\log \tilde w_{\ell,t} \le \log w_\ell + \frac{4\rad_{\ell,t}}{3}\) for \(\ell \not \in \mathcal{L}(\hat k_t)\),
and~\eqref{eq:case-2-to-gap} is due to the definition of \(\Delta_{\ell'}\) and the condition \(\rad_{\ell',t} \le \frac{\Delta_{\ell'}}{6 \Lmax}\).
The inequality \(\tilde u_t(k^*) > \tilde u_t (\tilde k_t)\) contradicts the definition of \(\tilde k_t\).

Both cases yield contradictions, which implies that the link \(\ell'\) is not selected at time slot \(t\) when \(\rad_{\ell',t} > \frac{\Delta_{\ell'}}{6 \Lmax}\). Substituting the definition of \(\rad_{\ell',t}\), we have \(N_{\ell',t} \le \frac{36\Lmaxsquare C}{\Delta_{\ell'}} \log\left( \frac{2Lt^3}{\delta} \right)\).
Therefore, the algorithm terminates within \(O\left( \Lmaxsquare H C \log \left( \frac{LH}{\delta} \right) \right)\) time slots and output the best path with probability at least \(1 - \delta\).

%% file: sections/proof-path.tex
\subsection{Proof for Algorithm~\ref{alg:path-level}}


\begin{proof}[Proof of Lemmas~\ref{lem:link-estimate} and~\ref{lem:link-estimate-qkd}]
    Assuming \(p^\pl(k) \in (\epsilon, 1)\) for some parameter \(\epsilon > 0\), we have \(\log p^\pl(k) \in (\log \epsilon, 1)\). Sampling this from the network benchmarking subroutine, we have the following lemma, which is equivalent to say that \(\log p^\pl(k)\) is \(\frac{1-\log \epsilon}{2}\)-sub-Gaussian. Then, with this sub-Gaussian property, we can devise the best arm identification for linear bandits algorithm as suggested by~\citep{tao2018best}.

    \citet[Corollary 16]{tao2018best} proves that, when the input sample size \(N \ge \frac{c_0(4L + (6+\epsilon)L^2)}{\epsilon^2}\log \frac{5L}{\delta}\), the estimates of \(\widehat{\log p_\ell}\) in Line~\ref{line:linear-regression} of Algorithm~\ref{alg:link-estimate} are accurate with probability at least \(1-\delta\), that is, \[
        \mathbb{P}\left( \abs{\bm x^T(k) \left(\log p_\ell - \widehat{\log p_{\ell}}\right)} \le \epsilon, \forall k\in\mathcal{K} \right) \ge 1 - \delta.
    \]

    \begin{lemma}\label{lma:scale-accuracy}
        For any two positive numbers \(0< a, b < 1\) and a parameter \(\epsilon > 0\),
        if \(\abs{\log a - \log b} \le \epsilon\),
        then we have \(\abs{\log \frac{2a + 1}{3} - \log \frac{2b+1}{3}} \le 2\epsilon\).
    \end{lemma}
    \begin{proof}[Proof of Lemma~\ref{lma:scale-accuracy}]

        We first show that \(
        e^{-\epsilon} \frac a b \le \frac{2a + 1}{2b + 1}
        \le e^{\epsilon} \frac a b,
        \)
        that is \(e^{-\epsilon}  \le \frac{2ab + b}{2ab + a}
        \le e^{\epsilon},\) which can be derived as follows, \begin{equation}
            \label{eq:2a-to-a}
            e^{-\epsilon} = \frac{1}{\epsilon} \le \frac{2a + 1}{2a + e^\epsilon} \overset{(a)}\le \frac{2ab + b}{2ab + a} \overset{(b)}\le \frac{2a+1}{2a+e^{-\epsilon}} \le \frac{1}{e^{-\epsilon}} = e^\epsilon,
        \end{equation}
        where inequalities (a) and (b) are due to \(\frac{a}{b} \in (e^{-\epsilon}, e^\epsilon)\).

        Then, we rearrange the inequality's LHS as follows,
        \[
            \log \frac{2a + 1}{3} - \log \frac{2b+1}{3}
            = \log \frac{2a + 1}{2b + 1}
            \overset{(a)}\in \left( -\epsilon + \log \frac a b, \epsilon + \log \frac a b \right)
            \overset{(b)}\subseteq (-2\epsilon, 2\epsilon),
        \]
        where inequality (a) is by~\eqref{eq:2a-to-a},
        and inequality (b) is  due to \(\frac{a}{b} \in (e^{-\epsilon}, e^\epsilon)\).
    \end{proof}

    With the above lemma, we have  \[
        \mathbb{P}\left( \abs{\bm x^T(k) \left(\log \frac{2p_\ell + 1}{3} - \log \left( \frac{2\log (\exp \widehat{\log p_{\ell,t}})+ 1}{2} \right)\right)} \le 2\epsilon, \forall k\in\mathcal{K} \right) \ge 1 - \delta.
    \]
\end{proof}